\documentclass[journal]{IEEEtran}
\ifCLASSINFOpdf
   \usepackage[pdftex]{graphicx}
\else
\fi
\usepackage{cite}

\usepackage{amsmath}

\usepackage{amsthm}

\usepackage{algorithm}

\usepackage{algpseudocode}

\usepackage{subfigure}
\usepackage{bm}

\newtheorem{define}{\textbf{Definition}}
\newtheorem{theorem}{\textbf{Theorem}}
\newtheorem{problem}{\textbf{Problem}}

\usepackage{url}

\begin{document}
\title{Pricing and Budget Allocation for IoT Blockchain with Edge Computing}

\author{Xingjian~Ding,
       Jianxiong~Guo,
        Deying~Li,
        and~Weili~Wu,~\IEEEmembership{Senior Member,~IEEE}% <-this % stops a space
\thanks{This work is supported by the National Natural Science Foundation of China (Grant NO.11671400, 61972404, 61672524), 
and partially by NSF 1907472.
%(\emph{Corresponding author: Deying Li.})
}% <-this % stops a space
\thanks{X. Ding and D. Li are with School of Information, Renmin University of China, Beijing, 100872, China (e-mail: dxj@ruc.edu.cn; deyingli@ruc.edu.cn).}
\thanks{D. Li is the corresponding author.}
\thanks{J. Guo and W. Wei are with the Department of Computer Science, Erik Jonsson
School of Engineering and Computer Science, University of Texas at Dallas,
Richardson, TX, 75080, USA (e-mail: jianxiong.guo@utdallas.edu; weiliwu@utdallas.edu).}% <-this % stops a space
}


\markboth{Journal of \LaTeX\ Class Files,~Vol.~14, No.~8, August~2015}%
{Shell \MakeLowercase{\textit{et al.}}: Bare Demo of IEEEtran.cls for IEEE Journals}



\maketitle

\begin{abstract}
Attracted by the inherent security and privacy protection of the blockchain, 
incorporating blockchain into Internet of Things (IoT) has been widely studied in these years.
However, the mining process requires high computational power, which prevents IoT devices from directly participating in 
blockchain construction. For this reason, edge computing service is introduced to help build the IoT blockchain, where
IoT devices could purchase computational resources from the edge servers. 
In this paper, we consider the case that IoT devices also have other tasks that need the help of edge servers, 
such as data analysis and data storage. 
The profits they can get from these tasks is closely related to the amounts of resources they purchased from the edge servers.
In this scenario, IoT devices will allocate their limited budgets to purchase different resources from different edge servers, 
such that their profits can be maximized.
Moreover, edge servers will set ``best'' prices such that they can get the biggest benefits.
Accordingly, there raise a pricing and budget allocation problem between 
edge servers and IoT devices. 
We model the interaction between edge servers and IoT devices as a multi-leader multi-follower Stackelberg game, whose
objective is to reach the Stackelberg Equilibrium (SE). 
We prove the existence and uniqueness of the SE point, and design efficient algorithms to reach the SE point.
In the end, we verify our model and algorithms by performing extensive simulations, and the results show the correctness and 
effectiveness of our designs.
\end{abstract}

\begin{IEEEkeywords}
Internet of things, Blockchain, Edge computing,  Stackelberg game.
\end{IEEEkeywords}

% \fi
\IEEEpeerreviewmaketitle

%%%%%%%%%%%%%%%%%%%%%%%%%%%%%%%%%%%%%%%%%%%%%%%%%%%%%%%%%%%%%%%%%%%%%%%%%%%%%%%%%%%%%


\section{introduction}
In the past few decades, the Internet of Things (IoT) has been greatly developed and attracted more and more attention in 
academia and industry. 
The IoT technology helps integrate data by connecting different types of devices and 
has played an irreplaceable role in many fields, such as smart homes, smart factories,  smart grids, and so on. 
In the traditional centralized IoT system, all IoT devices are connected to a centralized cloud server, which is used to manage
devices and coordinate communications among devices. 
The most serious drawback of this centralized architecture is that it faces many problems, such as single point of failure, 
poor scalability, and network congestion \cite{rehman2019cloud}. 
Some studies introduce distributed IoT \cite{roman2013features} and peer-to-peer (P2P) networks \cite{kim2020p2p} %kim2018efficient
 to overcome these problems. However, the above studies didn't solve the inherent threats and vulnerabilities of the IoT, such as security 
and privacy issues \cite{alaba2017internet}.

A very effective way to solve the above issues is to incorporate blockchain into IoT \cite{kshetri2017can}. 
The blockchain technology has been widely used since it was first implemented for Bitcoin in 2009 \cite{nakamoto2019bitcoin}.
Blockchain records data as a decentralized public ledger, it does not require a third party server to store the data. 
Instead, data are stored in the form of blocks and maintained by all of the members of the blockchain network. 
The distributed feature allows blockchain to avoid suffering single point of failure which may happen in centralized systems.
The blocks are linked by cryptography, and thus any change in a block will affect the subsequent blocks. 
The security of a blockchain mainly comes from the way that a new block is generated, which is called \emph{mining}.
To generate a new block, the members of the blockchain network need to win the competition of solving a hash puzzle which is 
very computation-consuming, and the winner will get a reward from the blockchain network platform. 
In this paper, we consider that the IoT blockchain network adopts the Proof-of-Work (PoW) consensus mechanism. 
It is worth mentioning that some researchers have proposed a Directed Acyclic Graph (DAG) based blockchain 
which is known as \emph{tangle} % \cite{popov2016tangle} 
for lightweight IoT applications,  such as IoTA \cite{divya2018iota}. However, the DAG-based
blockchain has many vulnerabilities, such as it facing the threats of denial of service attacks and spam attacks \cite{bai2018state}, 
and it's vulnerable to double-spending attacks \cite{ferraro2019stability}. Therefore, we do not adopt the DAG-based blockchain in this 
paper. 

As mentioned before, solving the hash puzzle is computation-consuming, so it's hard for the lightweight IoT devices to participate in
the mining process. Fortunately, edge computing service is helpful for establishing an IoT blockchain \cite{pan2018edgechain}, 
where IoT devices could purchase computational power from edge servers.
Consider such an IoT blockchain network that contains lots of IoT systems such as smart factories or smart homes. Each IoT system 
can be seen as a group, and all of the groups together maintain the operation of the blockchain. 
The IoT blockchain network will attract nearby IoT systems to join in it due to its security and privacy protection.
Motivated by the reward from the blockchain network platform,
the IoT devices in an IoT system will purchase the  computational resource from the edge server which provides hash computing service
(\emph{hash-server}) 
to participate in the mining process.
In addition, these IoT devices may have other tasks that require the help of the edge server which provides task processing service
(\emph{task-server}).
For example, IoT devices that used for building smart cities or realizing augmented reality (AR) need to store and process large 
amounts of data \cite{morabito2018consolidate}, which is very difficult for lightweight IoT devices to accomplish. 
Thus it's necessary for these devices to purchase resources from the task-server, so that they can perform their tasks with the help of 
the task-server.
Generally, the more resources they purchase, the faster and better they perform their
tasks, and the more profits they could get from the tasks. 
As IoT devices usually have limited budgets, how to allocate the budgets to purchase different resources from the two kind of servers 
so as to maximize the profits,
therefore, is an important problem for these IoT devices.

Driven by profit, edge servers will set the unit price of their resources to maximize their utilities, and accordingly, 
there raise a pricing and budget allocation problem between edge servers and IoT devices, 
We model the interaction between edge servers and IoT devices as a multi-leader multi-follower Stackelberg game, where edge servers
are leaders and IoT devices are followers. 
The main contributions of this paper are summarized as follows.

\begin{itemize}
	\item We introduce the IoT blockchain network with edge computing, and describe the operation of the IoT blockchain system.
	\item We establish a multi-leader multi-follower Stackelberg game to model the interaction between edge servers and IoT devices.
	We prove that the Stackelberg equilibrium of the game exists and is unique, and then propose algorithms to find the Stackelberg 
	equilibrium in limited interactions. 
	\item We perform extensive simulations to validate the feasibility and effectiveness of our proposed algorithms, simulation results
	show that our algorithms can quickly reach the unique Stackelberg equilibrium point.
\end{itemize}

The rest of this paper is structured as follows. Section \ref{Sec:related} introduces the related works. Section \ref{sec:Blockchain}
describes the IoT blockchain with edge computing. Section \ref{Sec:problem} presents the Stackelberg game.
Section \ref{sec:Game analysis} designs algorithms to get the Stackelberg equilibrium. 
Section \ref{sec:Simulation} performs numerical simulations. 
And finally, Section \ref{sec:Conclusions} concludes this paper.

%%%%%%%%%%%%%%%%%%%%%%%%%%%%%%%%%%%%%%%%%%%%%%%%%%%%%
%%%%%%%%%%%%%%%%%%%%%%%%%%%%%%%%
\section{related works} \label{Sec:related}
Due to the inherent security and privacy protection properties of the blockchain, 
incorporating blockchain technology into IoT has been widely studied in recent years. 
Novo \cite{novo2018blockchain} design an architecture for scalable access management in IoT based on blockchain technology.
To address the privacy and security issues in the smart grid, Gai \emph{et al.} \cite{gai2019permissioned} present a permissioned 
blockchain edge model by combing blockchain and edge computing technologies. 
Guo \emph{et al.} \cite{guo2020combined} design a blockchain-enabled energy management system to ensure the security of
energy trading between the power grid and energy stations. 
Li \emph{et al.} \cite{li2020resource} propose a resource optimization for delay-tolerant data in blockchain-enabled IoT, they use the 
blockchain technology to improve the data security and efficiency in the IoT system. 
Liu \emph{et al.} \cite{liu2020secure} propose a blockchain-based approach for the data provenance in IoT, which ensures the 
correctness and integrity of the query results.
Qi \emph{et al.} \cite{qi2020cpds} build a compressed and data sharing framework with the help of blockchain
technology, which provides efficient and private data management for industrial IoT.
Lei \emph{et al.} \cite{lei2020groupchain} design the \emph{groupchain} which is a two-chain structured blockchain to ensure the scalability 
of the IoT services with fog computing. 

There are some works that use the Stackelberg game to study the interaction among the participators in the edge computing-based 
blockchain network, which are closely related to our work.
Chang \emph{et al.} \cite{chang2020incentive} study the incentive mechanism for edge computing-based blockchain networks, in which 
they aim to find the Stackelberg equilibrium between the edge service provider and the miners. 
Yao \emph{et al.} \cite{yao2019resource} use a Stackelberg game to model the pricing and resource trading problem between the cloud 
provider and industrial IoT devices, and they find the near-optimal policy through a multiagent reinforcement learning algorithm.
Xiong \emph{et al.} \cite{xiong2017edge,xiong2018mobile} formulate a Stackelberg game to jointly maximize the profit of mobile 
devices and the edge server in mobile blockchain networks. 
Ding \emph{et al.} \cite{ding2020incentive} investigate the interaction between the blockchain platform and IoT devices, 
%where blockchain platform set the mining reward and IoT devices purchase resource from edge servers, 
where their objective is to find the Stackelberg equilibrium such that both the blockchain platform and IoT devices could maximize their 
utility and profits respectively.
Guo \emph{et al.} \cite{guo2020blockchain} study a Stackelberg game and double auction based task offloading scheme for mobile 
blockchain. 
However, all of these existing works only considered the computational power demand of IoT devices, 
%they didn't consider that IoT devices need to purchase other resources from edge servers. Besides, 
and the game models in these works only have one leader, which is fundamentally different from our work. 
%%%%%%%%%%%%%%%%%%%%%%%%%%%%%%%%%%%%%%%%%%%%%%%%%%%%%%
%%%%%%%%%%%%%%%%%%%%%%%%%%%%%%%

\section{IoT Blockchain with Edge Computing}\label{sec:Blockchain}
In this section, we introduce the model of the IoT blockchain with edge computing, and describe the operation of the 
blockchain system. Moreover, we analyze the security and reliability of the IoT blockchain network.
%%===============================================================

\begin{figure}
	\centering
	\includegraphics[width=.49\textwidth]{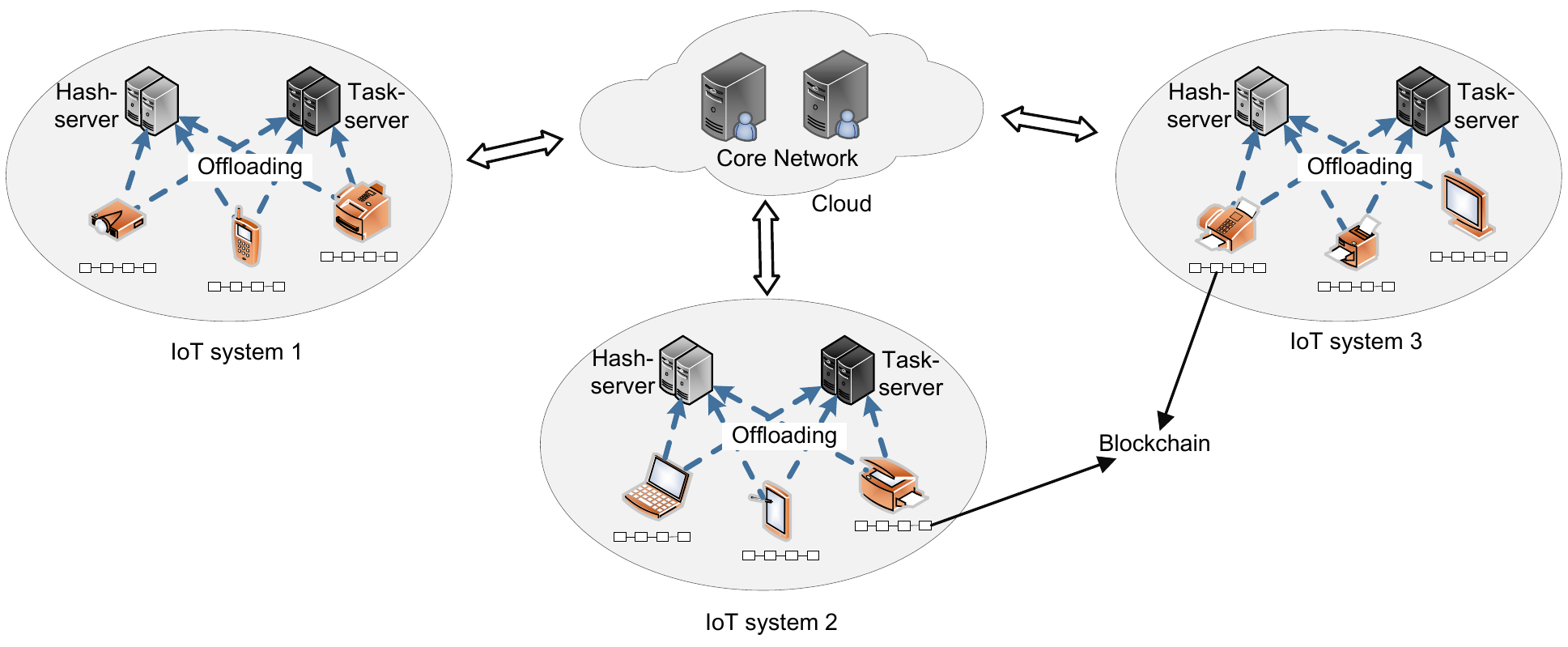}
	\caption{The architecture of the IoT blockchain with edge computing.}
	\label{fig:blockchain}
\end{figure}

\subsection{System Model} \label{sec:Model}
Fig. \ref{fig:blockchain} depicts the architecture of the system model of this paper.
Consider that there is an IoT blockchain network that has been running for a period which adopts the proof-of-work consensus
mechanism. The IoT blockchain network consists of lots of IoT systems such as smart factories or smart homes.
Each IoT system includes a set of IoT devices and can be seen as a group.
Due to the security and privacy protection brought by the blockchain, the IoT blockchain network will continuously attract other
IoT systems to join.

In each IoT system, there are two edge servers which provide hash computing service (\emph{hash-server}) and task processing 
service (\emph{task-server}), respectively.
Motivated by the reward from the blockchain network platform, devices in the IoT system would like to be miners of the blockchain 
network, that is, they will compete with other miners to scramble the right of generating a new block by solving a hash puzzle. 
Due to the limited computational power of these devices, they will purchase computational resources from the hash-server and then 
offload their hash puzzle to the hash-server during the mining process. 
Moreover, each IoT device has its own tasks, such as data collecting, data analysis, and data processing. IoT devices could benefit 
from performing these tasks. However, when the amount of data is relatively large, it is difficult for these IoT devices to perform the 
tasks. 
Then IoT devices will purchase task processing resources from the task-server to perform their
tasks. Generally, the more resource they purchase, the faster and better they perform their tasks, and then the more benefit they get 
from these tasks.

%%===============================================================
%%===============================================================
\subsection{Blockchain System}
\subsubsection{System Initialization}
Before each IoT device joins in the blockchain network, it needs to register with the Authentication Server (AS) which is authorized by
the blockchain platform. The process is as follows. 
A device $s_n$ first selects its own identifier $ID_n$, and then generate its 
public/private key pair $(PK_n, SK_n)$ with Elliptic Curve Digital Signature Algorithm (ECDSA) asymmetric cryptography 
\cite{johnson2001elliptic}. The public key is opened to the whole blockchain network, and the private key
is only known by the device itself. To avoid public key replacement attacks, each IoT device $s_n$ needs 
to get its digital certificate $Cert_n$ from the Certificate Authority (CA). The digital certificate $Cert_n$ is generated by CA's private key
$SK^{ca}$ with $s_n$'s identifier $ID_n$ and public key $PK_n$, i.e., $Cert_n = SK^{ca}(ID_n, PK_n)$. CA send the encrypted
message $m = PK_n(Cert_n)$ to $s_n$, and $s_n$ decrypt the message with its private key to get its digital certificate, that is, 
$Cert_n = SK_n(m)$. The digital certificate is used to uniquely identify the IoT device. 
After gets its digital certificate, the device $s_n$ submits its registration information $reg_n = PK^{as}(ID_n, PK_n, Cert_n)$ to the AS.
Upon receiving $reg_n$, AS get the registration information by decrypting $reg_n$ with its private key 
$(ID_n, PK_n, Cert_n) = SK^{as}(reg_n)$. Then AS check the identity of $s_n$, if $(ID_n, PK_n) = PK^{ca}(Cert_n)$, 
$s_n$ passes the authentication. Device $s_n$ gets its wallet address $WAD_n$ from AS, which is generated from its public key with 
SHA256,  RIPEMD-160 and BASE58 algorithms \cite{aitzhan2016security}. The AS will store the information 
$(ID_n, PK_n, Cert_n, WAD_n)$ about IoT device $s_n$.

\subsubsection{Create Transactions}
In the IoT blockchain network, devices can trading with each other, and purchase or exchange sensing data with each other. 
For example, a device $s_n$ wants to purchase the sensing data from $s_m$, $s_n$ first generate a request 
$req = \{ID_n || WAD_n || DataMes_n || T_{stamp}\}$, where $DataMes_n$ is the description of its request and $T_{stamp}$ is 
the timestamp of the message. 
Then $s_n$ signs the request message with its private key $SK_n$ and digital certificate $Cert_n$ to prove its identity, 
$Sreq = \{req || Sign_{SK_n}(Hash(req)) || Cert_n \}$.
At last, $s_n$ encrypt the message with $s_m$'s public key $EncSreq = PK_m(Sreq)$, and sent the message $EncSreq$ to $s_m$.
Upon receiving the message, $s_m$ first decrypt the message with its private key $Sreq = SK_m(EncSreq)$. 
Then check the signature and digital certificate of $s_n$. If $Hash(req) = PK_n(Sign_{SK_n}(Hash(req)))$, it can make sure that 
the message is sent by $s_n$. The digital certificate is used to validate the authenticity of the signature. $s_m$ use $CA$'s public
key to decrypt the digital certificate $Cert_n$, if $(ID_n, PK_n) = PK^{ca}(Cert_n)$, it's known that the signature is signed by $s_n$.
Then $s_m$ will send a response message to $s_n$, similar to the above process, all the messages between $s_n$ and $s_m$
are sent in an encrypted way. 
After finish the trading, $s_n$ will broadcast their trading record to the blockchain network, and waiting to be stored in the blockchain.
Besides the trading information between devices, some devices may want to store the sensing data which is important and sensitive into
blockchain. Both the trading records and sensing data are considered as transactions. 

\subsubsection{Building Blocks}
IoT devices collect a certain number of transactions in a period, and package them into a block. Each block is composed of two parts:
block content and block header. 
The block content records the detail of transactions in a Merkle tree structure.
The block header consists of the previous block hash value, which is used as a cryptographic link that 
creates the chain, a version number that used for tracking for software or protocol updates, a timestamp that records the time at which 
the block is generated, a Merkle tree root of all the transactions, a hash threshold value that records the current mining difficulty, and a 
nonce, which is used for solving the PoW puzzle. 
The mining process is similar to that in the Bitcoin system. Denoted $h_{data}$ by the block header excludes the nonce, then the mining 
process is to find a nonce $a$ such that $Hash(h_{data} + a) < difficulty$ \cite{nakamoto2019bitcoin}, where $difficulty$ is a 256-bits
binary number and is controlled by the blockchain platform to adjust the block generation speed.
As the hash operation is very costly, the hash task of each IoT device will be offloaded to the hash-server, and each device is only 
responsible for generating new blocks after receiving the result from the hash-server.
 
\subsubsection{Carrying Out Consensus Process}
The device who first solves the PoW puzzle gets the right to generate a new block, and then the new block needs to be verified by
other devices. By adopting the group signature and authentication scheme proposed in \cite{zhang2019group}, each IoT system 
in the blockchain network can be seen as a group. For a new block which is generated by a device in group $i$,
it needs to pass a two-round validation before to be added in the blockchain. 
In the first round, the block is checked by the devices in group $i$, each device will validate the transactions recorded in this block.
The block will get a signature if it passes the validation from a device, and
it can be broadcast to other groups for a second round validation only if it gets all the signatures of devices in group $i$.
In the second round, upon receiving the block, devices in other groups only check the signature attached in the block.
If more than $50\%$ of the devices agreed with the block, the block will be added into the blockchain.
Due to the constraints in memory, we let each IoT device only stores a certain number of the latest blocks, which is also applied in
\cite{koshy2020sliding,yang2020ldv}. The whole blockchain is stored in the monitoring nodes \cite{yang2020ldv} in each group 
(IoT system), as the monitoring nodes are authoritative and have larger memory capacity.

\subsection{Security and Reliability Analysis}
Different from traditional IoT systems, merging the IoT system into a blockchain network has many advantages, especially in terms of
security and reliability. Specifically, the IoT blockchain network with edge computing inherits the security and reliability performance 
of the blockchain, as shown in follows.

\begin{enumerate}
	\item \emph{Get rid of a third party:} IoT devices carry out transactions in a P2P manner, in which each device has the same 
	rights, and the trust between devices is built 
	with the help of the smart contract of the blockchain. Therefore, the IoT blockchain network guarantees the robustness and 
	scalability without involving a trusted third party.
	\item \emph{Privacy protection:} Blockchain can help deal with the increasing risk of sensitive data being exposed to malwares.
	In the blockchain network, the communication between devices uses the asymmetric encryption technology to protect the 
	sensitive data.
	%Each IoT device has a public/private key pair, where the public key is known by the whole network, and the private
	%key is kept by itself. 
	The message sent to a device will be encrypted with the receiver's public key, and it can only be decrypted 
	by the receiver's private key. In this way, even if malicious devices intercept the message, they cannot know the content of the 
	message. 
	%Moreover, the message sender will sign the message with its private key, and the receiver decrypted the signature 
	%with the sender's public key so that it can make sure that the message comes from the sender. The digital certificate is used 
	%to ensure that the public key of the sender will not be forged by attackers. The signature and the digital certificate together
	%guarantee that the malicious devices cannot send a message by impersonating others.
	\item \emph{Integrity:} The blocks are duplicated recorded in different devices in a distributed way, so it's very hard for attackers 
	to tamper with the blockchain. Besides, the blocks are linked together through cryptography. 
	If an attacker attempts to tamper with the transactions in a block, the hash value of each 
	subsequent block will be changed, and thus the attacker needs to redo the PoW puzzle of each subsequent block, 
	which is nearly impossible for the attackers. 
	%That means the attacker needs to fork the chain and create a longer sub-chain. It's nearly impossible for the attacker to achieve 
	%this as it's very hard for the attacker to control the majority of computation power. Thus the integrity of data can be guaranteed.
	\item \emph{Authentication:} In this IoT blockchain network, each new block needs to pass a two-round validation before it is 
	added into the blockchain. It's very hard for an attacker to control a whole group (IoT system), so the new block that contains 
	illegal transactions cannot pass the first round validation. Even if some attackers forge the signature of a group, the illegal block
	cannot pass the second round validation. 
\end{enumerate}

%%===============================================================
%%===============================================================
\section{multi-leader multi-follower Stackelberg game} \label{Sec:problem}
Consider that a new IoT system now join in the IoT blockchain network, the IoT devices in this system will purchase resources
from the edge servers to participate in the mining process and perform their tasks.
Driven by profit, 
the hash-server and task-server will adjust the unit price of their resources to maximize their utilities.
After the two servers publish their pricing strategies, 
each IoT device will determine its strategy for purchasing resources from the two servers according to the resource price and 
their budgets, such that their profits can be maximized.
In this section, we first give the utility functions of the two servers and the profit function of each device, and then 
describe the problem to be addressed in this paper, specifically, we model the 
interaction between the two servers and IoT devices as a multi-leader multi-follower Stackelberg game.

\subsection{Utility Function} \label{prob:def}

We assume that each IoT device has a unique budget to purchase resources from edge servers. The amount of resources
they purchase from the two edge servers depends on how many profits they can get from the trading and are limited by their budgets.
Each IoT device will allocate their budget to purchase different services from the hash-server and the task-server to maximize their 
profits. 
For the hash-server and the task-server, they will set the unit price of their resources to maximize their utilities. Moreover, there is 
a competition between the two servers. For example, if the unit price of resources from hash-server is too high, IoT devices would
purchase more resources from the task-server, and vice versa. 
Naturally, we 
model the interaction between the two servers and IoT devices as a multi-leader multi-follower Stackelberg game, where the hash-server
and the task-server act as leaders who first set the unit price of their resources, and IoT devices act as followers who determine their
strategies according to the leaders' bids.

We use $\mathcal{S} = \{s_1, s_2, \dots, s_n\}$ to denote the set of IoT devices, the budget for each device $s_i \in \mathcal{S}$ to 
purchase resources is $b_i$ where $b_i > 0$. The unit price of resources from the hash-server and the task-server are denoted by 
$p_h$ and $p_t$ per day, respectively. 
Let $x_i^h$ and $x_i^t$ be the amount of resources purchased by $s_i$ from the hash-server and the 
task-server, respectively. The amount of resources purchased buy each device is limited by its budget, that is, 
$x_i^h * p_h + x_i^t * p_t \leq b_i, \forall s_i \in \mathcal{S}$.

The profits of each IoT device $s_i \in \mathcal{S}$ includes two parts. The first part comes from mining new blocks for the blockchain 
network, which is related to the amount of resources $x_i^h$ purchased by $s_i$ from the hash-server.
The second part comes from performing tasks, which is related to the amount of resources $x_i^t$ purchased by $s_i$ from the 
task-server. We use $P_i^h$ and $P_i^t$ to denote the two parts of profits, respectively. In the following, we will describe 
how to calculate the two parts of profits.

As the blockchain network has been running for a period, in this paper, we assume that the total hash computational power of the 
blockchain network in a period time in the future can be estimated and is denoted by $H$. In the PoW consensus, the first miner who
solves the hash puzzle has the right to generate a new block and will get the reward from the blockchain network. The probability of a 
miner winning the mining competition is directly related to the hash computational power. 
We use $pro_i$ to denote the probability that device $s_i \in \mathcal{S}$ is the first one to solve the hash puzzle,
and $pro_i$ can be estimated by 
\begin{gather}
	pro_i = \frac{x_i^h}{H+x_i^h}.
\end{gather}

Generally, the blockchain network will adjust the difficulty of the hash puzzle periodically according to the total hash computation power
in the network to stabilize the block generation speed. We assume that an average of $N$ new blocks are generated per day, and 
the miners will get a reward $R$ for generating a new block. The expected reward obtained by device $s_i$ in per day is $pro_i RN$,
and the cost is $x_i^hp_h$.
Then the expected profits that $s_i$ gets from the mining process in a day is calculated as
\begin{gather}
	P_i^h = pro_i RN - x_i^hp_h.
\end{gather}

The profits $P_i^t$ got by $s_i$ comes from performing tasks is related to the amount of resources purchased by $s_i$ from the task 
server, we use a logarithmic function to estimate the benefits of device $s_i$ for performing tasks, and then $P_i^t$ is calculated as
\begin{gather}
	P_i^t = \alpha \log(1 + \beta x_i^t) - x_i^tp_t,
\end{gather}
where $\alpha > 0$ and $\beta \geq 1$ are two constant parameters. 

Therefore, the total profits got by device $s_i$ is calculated as
\begin{gather} \label{equ:miner_profits}
	\begin{split}
		P_i &= P_i^h + P_i^t  \\
		&= RN\frac{x_i^h}{H+x_i^h} - x_i^hp_h + \alpha \log(1 +  \beta x_i^t) - x_i^tp_t.
	\end{split}
\end{gather}

We use $U_h$ and $U_t$ to denote the utility of the hash-server and task-server, respectively. Assume that the unit
hash resource cost of  the hash-server is $c_h$, and the unit task resource cost of the task-server is $c_t$. Then
the utilities of the two servers can be calculated as
\begin{gather} \label{leaderU:hash}
	U_h = \sum_{s_i\in \mathcal{S}} (p_h - c_h)x_i^h, 
\end{gather}
\begin{gather}\label{leaderU:task}
	U_t =  \sum_{s_i\in \mathcal{S}} (p_t - c_t)x_i^t.
\end{gather}

The profits $P_i$ gots by device $s_i$ involves two parts, i.e., $P_i^h$ and $P_i^t$, which comes from trading with the two servers,
respectively. 
The first-order derivatives of $P_i^h$ and $P_i^t$ are
\begin{gather}
	\frac{\partial P_i^h}{\partial x_i^h} = RN\frac{H}{(H+x_i^h)^2} - p_h,
\end{gather}
\begin{gather}
	\frac{\partial P_i^t}{\partial x_i^t} = \frac{\alpha \beta}{1 + \beta x_i^t} - p_t.
\end{gather}

To make the problem reasonable, the conditions $\frac{\partial P_i^h}{\partial x_i^h} (0) \geq 0$ and 
$\frac{\partial P_i^t}{\partial x_i^t}(0) \geq 0$ should be hold, otherwise, IoT devices will never purchase resources from the hash-server 
or the task-server.  Hence, we have $p_h \leq \frac{RN}{H}$ and $p_t \leq \alpha \beta$. As servers will never sell their resources at a 
price below the cost, that is, $p_h \geq c_h$ and $p_t \geq c_t$. Therefore, in this paper, we assume that $c_h \leq p_h \leq \frac{RN}{H}$ and 
$c_t \leq p_t \leq \alpha \beta$.

\subsection{Problem Formulation}

The interaction between the two servers and IoT devices has two stages. In the upper stage, the hash-server and the task-server
offer a unit price of their resources. In the lower stage, IoT devices determine their strategies to maximize their profits according 
to the price of different services. In the following, we give a detailed definition of the problem in each stage.

\begin{problem}
	The problem in the lower stage (followers side).
	\begin{align}
		 &\max\limits_{x_i^h, x_i^t} ~~~ P_i\\
    		&s.t.  ~~~ x_i^h p_h + x_i^t p_t \leq b_i,  \label{equ:P1_c1} \\
		&~~~~~~~ x_i^h \geq 0, x_i^t \geq 0.   \label{equ:P1_c2}
	\end{align}
\end{problem}

\begin{problem} \label{pro:2}
	The problem in the upper stage (leaders side).
	\begin{align}
		 &\max\limits_{p_h} ~~~ U_h\\
    		&~ s.t.  ~~~c_h \leq p_h \leq \frac{RN}{H},
	\end{align}
	and 
	\begin{align}
		 &\max\limits_{p_t} ~~~ U_t\\
    		&~ s.t.  ~~~ c_t \leq p_t \leq \alpha \beta.
	\end{align}
\end{problem}

Note that in the lower stage, each IoT device make their decision independently, and in the upper stage, the two servers are also
non-cooperative. Therefore, the problems in the two stages form a non-cooperative multi-leader multi-follower Stackelberg game.
Our objective is to find the Stackelberg equilibrium (SE) point of the game, where none of the players of the game wants to change 
its strategy unilaterally. 
The SE point in our model is defined as follows.

\begin{define} \label{def:stack}
	Let $\bm{x_i^*} = \{x_i^{h*}, x_i^{t*}\}$ be a strategy of IoT device $s_i$, and we use
	$\bm{X^*} = \{\bm{x_1^*},\bm{x_2^*},\dots,\bm{x_n^*}\}$ to denote the set of strategies of all of the IoT devices. 
	Let $p_h^*$ and $p_t^*$ be the strategies of the hash-server and the task-server, respectively. 
	The point $(\bm{X^*}, p_h^*, p_t^*)$ is the Stackelberg equilibrium point if the following conditions are satisfied.
	\begin{align}
		%\begin{split}
			P_i(\bm{x_i^*}, p_h^*, p_t^*) &\geq P_i(\bm{x_i}, p_h^*, p_t^*), \forall s_i\in \mathcal{S}, \\
			U_h(p_h^*, p_t^*) &\geq U_h(p_h, p_t^*),\\
			U_t(p_t^*, p_h^*) &\geq U_t(p_t, p_h^*),
		%\end{split}
	\end{align}
	where $\bm{x_i} = \{x_i^h, x_i^t\}$ is an arbitrary feasible strategy for any device $s_i \in \mathcal{S}$, $p_h$ and $p_t$ are 
	arbitrary feasible strategies for the hash-server and the task-server, respectively.
\end{define}

%%%%%%%%%%%%%%%%%%%%%%%%%%%%%%%%%%%%%%%%%%%%%%%%%%%%%%%%%%%%%%%%%%%%%%%%%%%%%%%%%%%%%

%%%%%%%%%%%%%%%%%%%%%%%%%%%%%%%%%%%%%%%%%%%%%%%%%%%%%%%%%%%%%%%%%%%%%%%%%%%%%%%%%%%%%
\section{Solution of the multi-leader multi-follower Stackelberg game}\label{sec:Game analysis}
In this section, we analyze the existence and uniqueness of the Stackelberg equilibrium point of the multi-leader multi-follower 
Stackelberg game.  We first analyze the lower stage of the game, where each follower purchases different resources from the two 
servers with a limited budget to maximize its profits. Then we analyze the upper stage of the game, where the two servers determine 
their pricing strategies to maximize their utilities.

\subsection{Lower stage (followers side) analysis} \label{subsec:miners}
The second-order derivatives of the profits function $P_i$ of device $s_i$ are
\begin{gather}
	\frac{\partial^2 P_i}{\partial (x_i^h)^2} = \frac{-2RN}{(H+x_i^h)^3} < 0,
\end{gather}
\begin{gather}
	\frac{\partial^2 P_i}{\partial (x_i^t)^2} = \frac{-\alpha \beta^2}{(1 + \beta x_i^t)^2} < 0,
\end{gather}
\begin{gather}
	\frac{\partial^2 P_i}{\partial x_i^h \partial x_i^t} = 0.
\end{gather}
Therefore, the profits function $P_i$ is strictly concave, and the problem for each device $s_i$ in the lower stage is actually a convex 
optimization problem. Sequentially, we use the Karush–Kuhn–Tucker (KKT) conditions to solve the problem.

Let $\lambda_1, \lambda_2$ and $\lambda_3$ be the Lagrange's multipliers that associated with conditions in Eqs. (\ref{equ:P1_c1}) and 
(\ref{equ:P1_c2}). Then we define the Lagrangian function as follows.
\begin{gather}
	L_i = P_i + \lambda_1(b_i - x_i^h p_h - x_i^t p_t ) + \lambda_2 x_i^h + \lambda_3 x_i^t.
\end{gather} 

The KKT conditions (including four group of conditions) of Problem 1 are listed as follows.

\emph{Stationarity conditions:}
\begin{flalign}
 	&~~~~~~~\frac{\partial L_i}{\partial x_i^h} = RN\frac{H}{(H+x_i^h)^2} - p_h - \lambda_1 p_h + \lambda_2 = 0, &\label{KKT:1}\\
	&~~~~~~~\frac{\partial L_i}{\partial x_i^t} = \frac{\alpha\beta}{1 + \beta x_i^t} - p_t - \lambda_1 p_t + \lambda_3 = 0.  &\label{KKT:2}
\end{flalign}

\emph{Primal feasibility conditions:}
\begin{flalign}
	&~~~~~~~b_i - x_i^h p_h - x_i^t p_t \geq 0, & \label{KKT:Primal}  \\
	&~~~~~~~x_i^h, x_i^t \geq 0. & \label{KKT:Pri2}
\end{flalign}

\emph{Dual feasibility conditions:}
\begin{flalign}
	&~~~~~~~\lambda_1, \lambda_2, \lambda_3 \geq 0. &  \label{KKT:dufea}
\end{flalign}

\emph{Complementary slackness conditions:}
\begin{flalign}
	&~~~~~~~\lambda_1(b_i - x_i^h p_h - x_i^t p_t ) = 0, &\label{KKT:dual} \\
	&~~~~~~~\lambda_2 x_i^h = 0, &\label{KKT:4}\\
	&~~~~~~~\lambda_3 x_i^t = 0. &\label{KKT:5}
\end{flalign}

%\begin{align}
%	&\frac{\partial L_i}{\partial x_i^h} = RN\frac{H}{(H+x_i^h)^2} - p_h - \lambda_1 p_h + \lambda_2 = 0, \label{KKT:1}\\
%	&\frac{\partial L_i}{\partial x_i^t} = \frac{\alpha}{\beta + x_i^t} - p_t - \lambda_1 p_t + \lambda_3 = 0,  \label{KKT:2}\\
%	&b_i - x_i^h p_h - x_i^t p_t \geq 0, \\
%	&\lambda_1(b_i - x_i^h p_h - x_i^t p_t ) = 0, \label{KKT:dual} \\
%	&\lambda_2 x_i^h = 0, \label{KKT:4}\\
%	&\lambda_3 x_i^t = 0, \label{KKT:5}\\
%	&\lambda_1, \lambda_2, \lambda_3, x_i^h, x_i^t \geq 0. \label{KKT:6}
%\end{align}

The optimal solution of the problem is taken in one of the following four cases.

\emph{(1) Case 1:} $x_i^h = x_i^t = 0$. According to the KKT condition (\ref{KKT:dual}), we have $\lambda_1 = 0$ as $b_i > 0$.
Substitute it into KKT conditions (\ref{KKT:1}) and (\ref{KKT:2}), we have $\lambda_2 = p_h - \frac{RN}{H}$ and 
$\lambda_3 = p_t - \alpha\beta$. As $p_h \leq \frac{RN}{H}$ and $p_t \leq \alpha\beta$ hold, we have $\lambda_2 \leq 0$ 
and $\lambda_3 \leq 0$. The KKT conditions can be satisfied only when $\lambda_2 = 0$ and $\lambda_3 = 0$. Thus we need to check
whether $p_h = \frac{RN}{H}$ and $p_t = \alpha\beta$ hold, if yes, the optimal solution is $x_i^h = x_i^t = 0$, otherwise,
the optimal solution is not in this case.

\emph{(2) Case 2:} $x_i^h = 0$ and $x_i^t > 0$.
In this case, we have $\lambda_3 = 0$ according to the KKT condition (\ref{KKT:5}). 

Consider $\lambda_1 = 0$, substitute $x_i^h = 0$ and $\lambda_1 = 0$ into (\ref{KKT:1}), we have $\lambda_2 = p_h - \frac{RN}{H} \leq 0$.
If $p_h < \frac{RN}{H}$, the KKT condition (\ref{KKT:dufea}) cannot be satisfied as $\lambda_2 < 0$, which means $\lambda_1 = 0$
is not feasible in this case. Otherwise, if  $p_h = \frac{RN}{H}$, we have $\lambda_2 = 0$.
Substitute $\lambda_1 =  \lambda_3 = 0$ into (\ref{KKT:2}), we have $x_i^t = \frac{\alpha}{p_t} - \frac{1}{\beta}$. 
Then we need to check whether the primal feasibility condition (\ref{KKT:Primal}) is satisfied, if yes, the optimal solution is 
$(0, \frac{\alpha}{p_t} - \frac{1}{\beta})$, if no, $\lambda_1 = 0$ is not feasible in this case.

Consider $\lambda_1 > 0$, according to (\ref{KKT:dual}), we have $b_i - x_i^h p_h - x_i^t p_t  = 0$. 
Combing $x_i^h = 0$, we have 
\begin{gather} \label{case3:s1}
	x_i^t = \frac{b_i}{p_t }.
\end{gather}
Substitute (\ref{case3:s1}) and $\lambda_3=0$ into KKT condition (\ref{KKT:2}),
we have
\begin{gather} \label{case3:lambda1}
	\lambda_1 = \frac{\alpha\beta}{\beta b_i + p_t} - 1.
\end{gather}
Substitute (\ref{case3:lambda1}) and $x_i^h = 0$ into (\ref{KKT:1}), we have
\begin{gather} \label{case3:lambda2}
	\lambda_2 = \frac{\alpha\beta p_h}{\beta b_i+  p_t} - \frac{RN}{H}.
\end{gather}
%Next, we will check whether the KKT conditions are satisfied. 
It obvious that $x_i^t = \frac{b_i}{p_t}> 0$ is satisfied. 
If $\lambda_1$ got by (\ref{case3:lambda1}) satisfies $\lambda_1 > 0$ and $\lambda_2$ got by (\ref{case3:lambda2}) satisfies 
$\lambda_2 \geq 0$, it means all of the KKT conditions can be satisfied, and thus the optimal solution can be determined as
$(0, \frac{b_i}{p_t})$.
Otherwise, the optimal solution is not in Case 2.

\emph{(3) Case 3:} $x_i^h > 0$ and $x_i^t = 0$.
In this case, we have $\lambda_2 = 0$ according to the KKT condition (\ref{KKT:4}).

Consider $\lambda_1 = 0$, substitute $x_i^t = 0$ and $\lambda_1 = 0$ into (\ref{KKT:2}), we have 
$\lambda_3 = p_t - \alpha\beta \leq 0$.
If $p_t < \alpha\beta$, the KKT condition (\ref{KKT:dufea}) cannot be satisfied as $\lambda_3 < 0$, which implies 
$\lambda_1 = 0$ in not feasible in this case.
Otherwise, if $p_t = \alpha\beta$, we have $\lambda_3 = 0$. Substitute $\lambda_1 = \lambda_2 = 0$ into (\ref{KKT:1}),
we have $x_i^h = \sqrt{\frac{RNH}{p_h}} - H$.
Then we need to check whether the primal feasibility condition (\ref{KKT:Primal}) is satisfied, if yes, the optimal solution is 
$(\sqrt{\frac{RNH}{p_h}} - H, 0)$, if no, $\lambda_1 = 0$ is not feasible in this case.

Consider $\lambda_1 > 0$, we have $b_i - x_i^h p_h - x_i^t p_t  = 0$ according to (\ref{KKT:dual}).
Combing $x_i^t = 0$, we have 
\begin{gather} \label{case4:xh}
	x_i^h = \frac{b_i}{p_h}.
\end{gather}
Substitute (\ref{case4:xh}) and $\lambda_2 = 0$ into KKT condition (\ref{KKT:1}), we have
\begin{gather} \label{case4:lambda1}
	\lambda_1 = \frac{RNHp_h}{(b_i + Hp_h)^2} - 1.
\end{gather}
Substitute (\ref{case4:lambda1}) and $x_i^t = 0$ into (\ref{KKT:2}), we have
\begin{gather} \label{case4:lambda3}
	\lambda_3 = \frac{RNHp_hp_t}{(b_i + Hp_h)^2} - \alpha\beta.
\end{gather}
It's obvious that $x_i^h = \frac{b_i}{p_h} > 0$ is satisfied. 
If $\lambda_1$ got by (\ref{case4:lambda1}) satisfies $\lambda_1 > 0$ and $\lambda_3$ got by (\ref{case4:lambda3})
satisfies $\lambda_3 \geq 0$, it means that all of the KKT conditions can be satisfied, and thus the optimal solution 
is $(\frac{b_i}{p_h}, 0)$. Otherwise, the optimal solution is not in Case 3.

\emph{(4) Case 4:} $x_i^h >0$ and $x_i^t > 0$. In this case, we have $\lambda_2 = \lambda_3 = 0$ according to the KKT conditions 
(\ref{KKT:4}) and (\ref{KKT:5}). 

Suppose $\lambda_1 = 0$, substitute it into the stationarity conditions (\ref{KKT:1}) and (\ref{KKT:2}), we have 
\begin{gather} \label{case2:solution1}
	x_i^h = \sqrt{\frac{RNH}{p_h}} - H; x_i^t = \frac{\alpha}{p_t} - \frac{1}{\beta}.
\end{gather}
As $p_h \leq \frac{RN}{H}$ and $p_t \leq \alpha\beta$, it's easy to know that the condition (\ref{KKT:Pri2}) is satisfied. Then we 
check whether the condition (\ref{KKT:Primal}) is satisfied. If yes, then the solution shown in Eq. (\ref{case2:solution1}) is the optimal
solution. Otherwise, $\lambda_1 = 0$ is not feasible in this case. 

Now we consider the case that $\lambda_1 > 0$.
According to the condition (\ref{KKT:dual}), we have $b_i - x_i^h p_h - x_i^t p_t  = 0$. 
Solving conditions (\ref{KKT:1}) and (\ref{KKT:2}), we have
\begin{gather} \label{case2:s2}
	x_i^h = \sqrt{\frac{RNH}{p_h + \lambda_1p_h}} - H; x_i^t = \frac{\alpha}{p_t + \lambda_1 p_t} - \frac{1}{\beta}.
\end{gather}
Substitute (\ref{case2:s2}) into $b_i - x_i^h p_h - x_i^t p_t  = 0$, we have
\begin{gather} \label{case2:con1}
	A - B\sqrt{\frac{1}{1+\lambda_1}} - \frac{\alpha}{1+ \lambda_1} = 0.
\end{gather}
where $A = b_i + Hp_h + \frac{1}{\beta} p_t > 0$, and $B=\sqrt{RNHp_h} > 0$. 
Let $t = \sqrt{1+ \lambda_1} > 0$, substitute $t$ into (\ref{case2:con1}) and solve the function, we have 
$t = \frac{B \pm \sqrt{B^2 + 4A \alpha}}{2A}$. As $t > 0$, we have
\begin{gather} \label{case2:t}
	t =  \sqrt{1+ \lambda_1} = \frac{B + \sqrt{B^2 + 4A \alpha}}{2A}.
\end{gather}
By solving (\ref{case2:t}), we have
\begin{gather} \label{case2:lambda}
	\lambda_1 = \left(\frac{B + \sqrt{B^2 + 4A \alpha}}{2A}\right)^2 - 1.
\end{gather}
Then we check whether $\lambda_1$ got by (\ref{case2:lambda}) satisfies $\lambda_1 > 0$. 
 If $\lambda_1 \leq 0$, it means that the solutions found in (\ref{case2:s2}) cannot satisfy all of the KKT conditions, and thus the optimal 
 solution is not in Case 4.
 Otherwise, we substitute (\ref{case2:lambda}) into (\ref{case2:s2}), and we will get a solution $(x_i^h, x_i^t)$. Then we check whether
 the solution $(x_i^h, x_i^t)$ satisfies $x_i^h >0$ and $x_i^t > 0$. 
 If yes, the solution $(x_i^h, x_i^t)$ %got by substituting (\ref{case2:lambda}) into (\ref{case2:s2}) 
 is the optimal solution. If no, the optimal solution is not in Case 4.

The optimal solution will be found in one of the four cases from Case 1 to Case 4. 
Based on above analysis, we present the algorithm $\sf{FOSD}$ to solve Problem 1 in the lower stage.
The pseudo-code is shown in Algorithm \ref{algorithm:miner}.

\begin{algorithm} 
	\caption{Find Optimal Strategy for a Device $s_i$. ($\sf{FOSD}$)}
	\label{algorithm:miner}
	\begin{algorithmic}[1]
		\Require H, R, N, $\alpha$, $\beta$, $b_i$, $p_h$ and $p_t$;
		\Ensure The optimal strategy $(x_i^h, x_i^t)$ for IoT device $s_i$;
		\Statex // \textbf{Case 1:}
		\If{$p_h = \frac{RN}{H}$ \&\& $p_t = \alpha\beta$}
			\State \Return $(0,0)$;
		\EndIf
		
		\Statex // \textbf{Case 2-1:}
		\State $x_i^t = \frac{\alpha}{p_t} - \frac{1}{\beta}$;
		\If{$p_h = \frac{RN}{H}$ \&\& $b_i - x_i^t p_t \geq 0$}
				\State \Return $(0, x_i^t)$;
		\EndIf
		\Statex // \textbf{Case 2-2:}
		\State $x_i^t = \frac{b_i}{p_t}$;
		 $\lambda_1 = \frac{\alpha\beta}{\beta b_i + p_t} - 1$; $\lambda_2 = \frac{\alpha\beta p_h}{\beta b_i+ p_t} - \frac{RN}{H}$;
		\If{$\lambda_1 > 0$ \&\& $\lambda_2 \geq 0 $}
				\State \Return $(0, x_i^t)$;
		\EndIf
		
		\Statex // \textbf{Case 3-1:}
		\State $x_i^h = \sqrt{\frac{RNH}{p_h}} - H$;
		\If{$p_t = \alpha\beta$ \&\& $b_i - x_i^h p_h \geq 0$}
				\State \Return $(x_i^h, 0)$;
		\EndIf
		\Statex // \textbf{Case 3-2:}
		\State $x_i^h = \frac{b_i}{p_h}$;
		 $\lambda_1 = \frac{RNHp_h}{(b_i + Hp_h)^2} - 1; \lambda_3 = \frac{RNHp_hp_t}{(b_i + Hp_h)^2} - \alpha\beta$;
		\If{$\lambda_1 > 0$ \&\& $\lambda_3 \geq 0$}
				\State \Return $(x_i^h, 0)$;
		\EndIf
		
		\Statex // \textbf{Case 4-1:}
		\State $x_i^h = \sqrt{\frac{RNH}{p_h}} - H; x_i^t = \frac{\alpha}{p_t} - \frac{1}{\beta}$;
		\If{$b_i - x_i^h p_h - x_i^t p_t \geq 0$}  
			\State \Return $(x_i^h, x_i^t)$;
		\EndIf
		\Statex // \textbf{Case 4-2:}
		\State $A = b_i + Hp_h + \frac{1}{\beta} p_t $, $B=\sqrt{RNHp_h} $;
		\State $\lambda_1 = \left(\frac{B + \sqrt{B^2 + 4A \alpha}}{2A}\right)^2 - 1$;
		\If{$\lambda_1 > 0$}
			\State $x_i^h = \sqrt{\frac{RNH}{p_h + \lambda_1p_h}} - H; x_i^t = \frac{\alpha}{p_t + \lambda_1 p_t} - \frac{1}{\beta}$;
			\If{$x_i^h > 0 $ \&\& $ x_i^t > 0$}
				\State \Return $(x_i^h, x_i^t)$;
			\EndIf
		\EndIf
		
	\end{algorithmic}
\end{algorithm}

%%%%%%%%%%%%%%%%%%%%%%%%%%%%%%%%%%%%%%%%%%%%%%%%%%%%%%%%%%%%%%%%%%%%%%%%%%%%%%%%%%%%%
\subsection{Upper stage (leaders side) analysis}
On the leaders' side, hash-server and task-server set their unit prices $p_h$ and $p_t$ of resources to maximize their utilities $U_h$ and
$U_t$ which are calculated by Eqs. (\ref{leaderU:hash}) and (\ref{leaderU:task}), respectively. Note that the pricing strategies of leaders
will directly affect the strategies of followers, which has been analyzed in Section \ref{subsec:miners}. Therefore, the strategies of
the hash-server and task-server will affect each other's utility. The game between the two servers is non-cooperative and competitive. 
To maximize their utilities, each server (leader) should give a suitable unit price of its resource. For example, for the hash-server, if 
the unit price of resource $p_h$ is too high, the followers will prefer to purchase more resources from the task-server, and thus the
hash-server will get a very low utility. On the contrary, if $p_h$ is set too low, although the followers tend to purchase resources from
the hash-server, the total purchased resources are limited due to the budget limitation of these followers, which also results in a low
utility.  
In the following, we will show that the game between the two servers will reach a Nash equilibrium, and we design an algorithm to
find the Nash equilibrium point.

The concept of the Nash Equilibrium (NE) of the game between the two servers is defined as follows.

\begin{define}
	Let $p_h^*$ and $p_t^*$ be the strategies of the hash-server and the task-server, respectively, ($p_h^*$, $p_t^*$) is the Nash 
	equilibrium if the following conditions are satisfied.
	\begin{align}
		U_h(p_h^*, p_t^*) &\geq U_h(p_h, p_t^*),\\
		U_t(p_t^*, p_h^*) &\geq U_t(p_t, p_h^*),
	\end{align}
	where $p_h$ and $p_t$ are arbitrary feasible strategies for the hash-server and the task-server, respectively.
\end{define}

According to the definition, at the NE point, none of the servers can improve its utility by unilaterally changing its strategy. Therefore, 
when the game reaches a NE point, the interaction between the two servers is suspended, and the pricing strategies of the two servers
will never change again. 
Next, we will prove the existence and uniqueness of the NE point of the game between the two servers.

\begin{theorem} \label{Theorem:NE}
 	The NE point of the game between the two servers exists and is unique. 
\end{theorem}
\begin{proof}
	As defined in Problem \ref{pro:2}, the strategy space of the two servers is $[c_h, \frac{RH}{H}]\times[c_t, \alpha\beta]$,
	which is a non-empty, closed and convex subset of the Euclidean space. 
	%According to the definitions of the utilities and the analysis in Section \ref{prob:def}, we have 
	%$U_h \geq 0 = U_h(c_h) = U_h(\frac{RN}{H})$ and $U_t \geq 0 = U_t(c_t) = U_t(\frac{\alpha}{\beta})$. 
	Next, we calculate the second order derivatives of utility functions $U_h(\cdot)$ and $U_t(\cdot)$.
	
	We first consider the utility of the hash-server. According analysis in Section \ref{subsec:miners}, the amount of purchased 
	resources $x_i^h$ of device $s_i$ from the hash-server must be one of the four following cases: 
	(1) $x_i^h = \sqrt{\frac{RNH}{p_h}} - H$; (2) $x_i^h = \sqrt{\frac{RNH}{p_h + \lambda_1p_h}} - H$, where $\lambda_1 = $
	(\ref{case2:lambda}); (3) $x_i^h = 0$; (4) $x_i^h = \frac{b_i}{p_h}$.
	Then the first order derivative of $x_i^h(\cdot)$ with respect to $p_h$ of this four cases is calculated as:
	\begin{align}
		\frac{\partial x_i^h}{\partial p_h} &= - \frac{1}{2} \sqrt{RNH}  \left(p_h\right)^{-\frac{3}{2}} \\
			&= - \frac{1}{2} \sqrt{\frac{RNH}{1 + \lambda_1}} \left(p_h\right)^{-\frac{3}{2}} , \lambda_1 = (\ref{case2:lambda}) \\
			& = 0 \\
			& = -b_i (p_h)^{-2}.
	\end{align}
	
	The second order derivative of  $x_i^h(\cdot)$ with respect to $p_h$ of this four cases is calculated as:
	\begin{align}
		\frac{\partial^2 x_i^h}{\partial (p_h)^2} &= \frac{3}{4} \sqrt{RNH}  \left(p_h\right)^{-\frac{5}{2}} \\
			&=  \frac{3}{4} \sqrt{\frac{RNH}{1 + \lambda_1}} \left(p_h\right)^{-\frac{5}{2}} , \lambda_1 = (\ref{case2:lambda}) \\
			& = 0 \\
			& = 2 b_i (p_h)^{-3}.
	\end{align}
	
	According to function (\ref{leaderU:hash}), the second order derivative of $U_h(\cdot)$ with respect to $p_h$ is
	\begin{gather}
		\frac{\partial ^2 U_h}{\partial (p_h)^2} = \sum\limits_{s_i \in \mathcal{S}} \left( 2 \frac{\partial x_i^h}{\partial p_h} 
		+ (p_h - c_h) \frac{\partial^2 x_i^h}{\partial (p_h)^2}  \right).
	\end{gather}
	As $p_h - c_h < p_h$,  combing the functions of $\frac{\partial x_i^h}{\partial p_h}$ and $\frac{\partial^2 x_i^h}{\partial (p_h)^2}$,
	we have $\frac{\partial ^2 U_h}{\partial (p_h)^2} \leq 0$. Therefore, the utility function $U_h(\cdot)$ of the hash-server is a 
	concave function with respect to $p_h$.
	
	Similarly, we have $\frac{\partial ^2 U_t}{\partial (p_t)^2} \leq 0$, and we know that the utility function $U_t(\cdot)$ of the 
	task-server is a concave function with respect to $p_t$. Thus the interaction between the two servers form a concave 2-person 
	game. According to \cite{rosen1965existence}, we know that the NE point of the game between the two servers exists and is 
	unique. 
\end{proof}

After the leaders (the two servers) issue their strategies, the followers (IoT devices) will determine their strategies for purchasing 
different resources from the two servers, as is discussed in Section \ref{subsec:miners}. The interaction between servers and 
IoT devices formulates a multi-leader multi-follower Stackelberg game, and the objective is to find the Stackelberg equilibrium (SE)
point of the game, which is defined as Definition \ref{def:stack}.
Next, we will prove that the Stackelberg equilibrium of the game between servers and IoT devices exists and is unique.

\begin{theorem} \label{theorem2}
The SE point of the game between servers and IoT devices exists and is unique.
\end{theorem}
\begin{proof}
	As analyzed in Section \ref{subsec:miners}, each IoT device will find its optimal strategy in one of the three cases from Case 2
	to Case 4, which indicates that the strategy of each IoT device is unique after the two servers give their pricing strategies. 
	According to Theorem \ref{Theorem:NE}, the game between the two servers has a unique NE point, we thus can conclude that 
	the SE point of the game between servers and IoT devices exists and is unique.
\end{proof}

To find the NE point of the game between the two servers, based on the optimal strategies of IoT devices, we proposed an algorithm
$\sf{FNES}$ to find the final pricing strategies of the two servers. The $\sf{FNES}$ is based on sub-gradient technique 
\cite{boyd2004convex,zhang2016multi}, and it invokes $\sf{FOSD}$ as its subroutine. 
The pseudo-code of algorithm $\sf{FNES}$ is shown in Algorithm \ref{algorithm:leader}.

\begin{algorithm} 
	\caption{Find Nash Equilibrium for Servers. ($\sf{FNES}$)}
	\label{algorithm:leader}
	\begin{algorithmic}[1]
		\Require The game between the two servers and IoT devices;
		\Ensure The pricing strategy $(p_h, p_t)$ for the two servers;
		\State Initialize: $p_h = \frac{1}{2}(c_h + \frac{RN}{H})$, $p_t = \frac{1}{2}(c_t + \alpha\beta)$;
		\State Set a small step $\Delta$, and the attenuation coefficient $\delta$ of the step;
		\While{true}
			\State $p_h' = p_h$, $p_t' = p_t$;
			\Statex // \textbf{Adjust strategy for the hash-server}
			\State Calculate $U_h(p_h, p_t)$, $U_h(p_h + \Delta, p_t)$ and $U_h(p_h - \Delta, p_t)$ by invoking 
			Algorithm $\sf{FOSD}$ for each IoT device with parameters $(p_h, p_t)$, $(p_h + \Delta, p_t)$ and 
			$(p_h - \Delta, p_t)$, respectively;
			\If{$U_h(p_h + \Delta, p_t) \geq U_h(p_h, p_t)$ \&\& $U_h(p_h + \Delta, p_t) \geq U_h(p_h - \Delta, p_t)$}
				\State $p_h = \min\{p_h + \Delta, \frac{RN}{H}\}$;
			\ElsIf{$U_h(p_h - \Delta, p_t) \geq U_h(p_h, p_t)$ \&\& $U_h(p_h - \Delta, p_t) \geq U_h(p_h + \Delta, p_t)$}
				\State $p_h = \max\{p_h - \Delta, c_h\}$;
			\EndIf
			
			\Statex // \textbf{Adjust strategy for the tash-server}
			\State Calculate $U_t(p_t, p_h)$, $U_t(p_t + \Delta, p_h)$ and $U_t(p_t - \Delta, p_h)$ by invoking 
			Algorithm $\sf{FOSD}$ for each IoT device with parameters $(p_t, p_h)$, $(p_t + \Delta, p_h)$ and 
			$(p_t - \Delta, p_h)$, respectively;
			\If{$U_t(p_t + \Delta, p_h) \geq U_t(p_t, p_h)$ \&\& $U_t(p_t + \Delta, p_h) \geq U_t(p_t - \Delta, p_h)$}
				\State $p_t = \min\{p_t + \Delta, \alpha\beta\}$;
			\ElsIf{$U_t(p_t - \Delta, p_h) \geq U_t(p_t, p_h)$ \&\& $U_t(p_t - \Delta, p_h) \geq U_t(p_t + \Delta, p_h)$}
				\State $p_t = \max\{p_t - \Delta, c_t\}$;
			\EndIf
			
			\If{$p_h' = p_h$ \&\& $p_t' = p_t$}
				\State Break;
			\EndIf
			
			\Statex // \textbf{Reduce the step}
			\State $\Delta = \delta \cdot \Delta$;
		\EndWhile \\
		\Return $(p_h, p_t)$
	\end{algorithmic}
\end{algorithm}

In $\sf{FNES}$, we first set a feasible pricing strategy for each server,  and a small step $\Delta$ is set to update the strategies 
of servers. We iteratively adjust the pricing strategy for each server in an alternative way. For the hash-server, we will calculate 
its utility $U_h$ with pricing strategies $p_h$, $p_h + \Delta$ and $p_h - \Delta$, and the best pricing strategy will be selected
as the latest strategy in the next round. Note that, the value of $U_h$ is related to the strategy of each IoT device, and thus we 
need to invoke the $\sf{FOSD}$ algorithm to get the strategy of each IoT device. 
The strategy adjustment of the task-server is similar to that of the hash-server.
In each iteration, we update the step $\Delta$ with the attenuation coefficient $\delta$, where $\delta \in (0,1)$.
The $\sf{FNES}$ terminates when none of the servers will change its pricing strategy.

%%%%%%%%%%%%%%%%%%%%%%%%%%%%%%%%%%%%%%%%%%%%%%%%%%%%%%%%%%%%%%%%%%%%%%%%%%%%%%%%%%%%%
%%%%%%%%%%%%%%%%%%%%%%%%%%%%%%%%%%%%%%%%%%%%%%%%%%%%%%%%%%%%%%%%%%%%%%%%%%%%%%%%%%%%%

\section{Simulations}\label{sec:Simulation}
In this section, we conduct numerical experiments to validate the feasibility and effectiveness of our algorithms.

%%================================================================
\subsection{Experimental settings}
In the experiments, we assume the total hash computational power $H$ of the IoT blockchain network in the next period is estimated to 
be $1000$ $GH/s$. The mining reward $R$ from the blockchain platform is set to be $300$, and the number of generated new blocks 
per day $N$ is set to be $144$. For the profit function of performing tasks, we set $\alpha$ to be $40$, and $\beta$ to be $2$. 
The unit resource cost of the hash-server and the task-server is set to be 10, that is, $c_h = c_t = 10$.
We consider there are $5$ IoT devices in the IoT system that
would like to purchase resources from the two servers, the budget of each device is
$50, 60, 70, 80$, and $90$, respectively.
For the algorithm $\sf{FNES}$, we set the set $\Delta$ to be $1$, and the attenuation coefficient $\delta$ to be $0.99$.
Unless other declared, the above are the default settings of our experiments.

\subsection{Results and Analyses}
\emph{1) The convergence of algorithm $\sf{FNES}$:}
In algorithm $\sf{FNES}$, we set the default initial value of the pricing strategies of the two servers as 
$p_h = \frac{1}{2}(c_h + \frac{RN}{H}) = 26.6$ and $p_t = \frac{1}{2}(c_t + \alpha\beta) = 45$. 
As shown in Fig. \ref{fig:itera-mid}, after $23$ iterations, the interaction between the two servers reach the Nash Equilibrium.
However, when we set the initial value of the pricing strategies of the two servers as
$p_h = \frac{RN}{H} = 43.2$ and $p_t = \alpha\beta=80$, as shown in Fig. \ref{fig:itera-high}, 
it needs about $130$ iterations to reach the Nash Equilibrium, which is much slower than that in Fig. \ref{fig:itera-mid}.
This indicates that different initialization will significantly affect the convergence speed of algorithm $\sf{FNES}$. 
We can also see that Fig. \ref{fig:itera-mid} and Fig. \ref{fig:itera-high} reach the same Nash Equilibrium even though the initialization
is different, which implies the correctness of our analysis in Theorem \ref{theorem2}.

We investigate the effect of the step $\Delta$ on the convergence of algorithm $\sf{FNES}$ in Fig. \ref{fig:itera-delta}.
It can be seen that it needs more iterations to reach the Nash Equilibrium when we adopt a smaller $\Delta$. 
However, although using a larger $\Delta$ can quickly approach the equilibrium point, we have to wait for more iterations if we want 
to get a more accurate solution, as we need to wait for the step $\Delta$ to decay to a sufficiently small level. 
Therefore, how to select a suitable value of $\Delta$ depends on the trade-off between the convergence speed and solution
accuracy. If we care more about the convergence speed other than the accuracy, we should adopt a relatively large $\Delta$, 
otherwise, we should adopt a small $\Delta$. 
There exists an alternative way to quickly reach a more accurate equilibrium point, that is, we use a large $\Delta$ to quickly 
approach the equilibrium point, and then set the $\Delta$ to a small value to improve the accuracy.

\begin{figure}
	\centering
	\subfigure[Initialization:($p_h$=26.6, $p_t$=45)]{ \label{fig:itera-mid}
		\includegraphics[width=.22\textwidth]{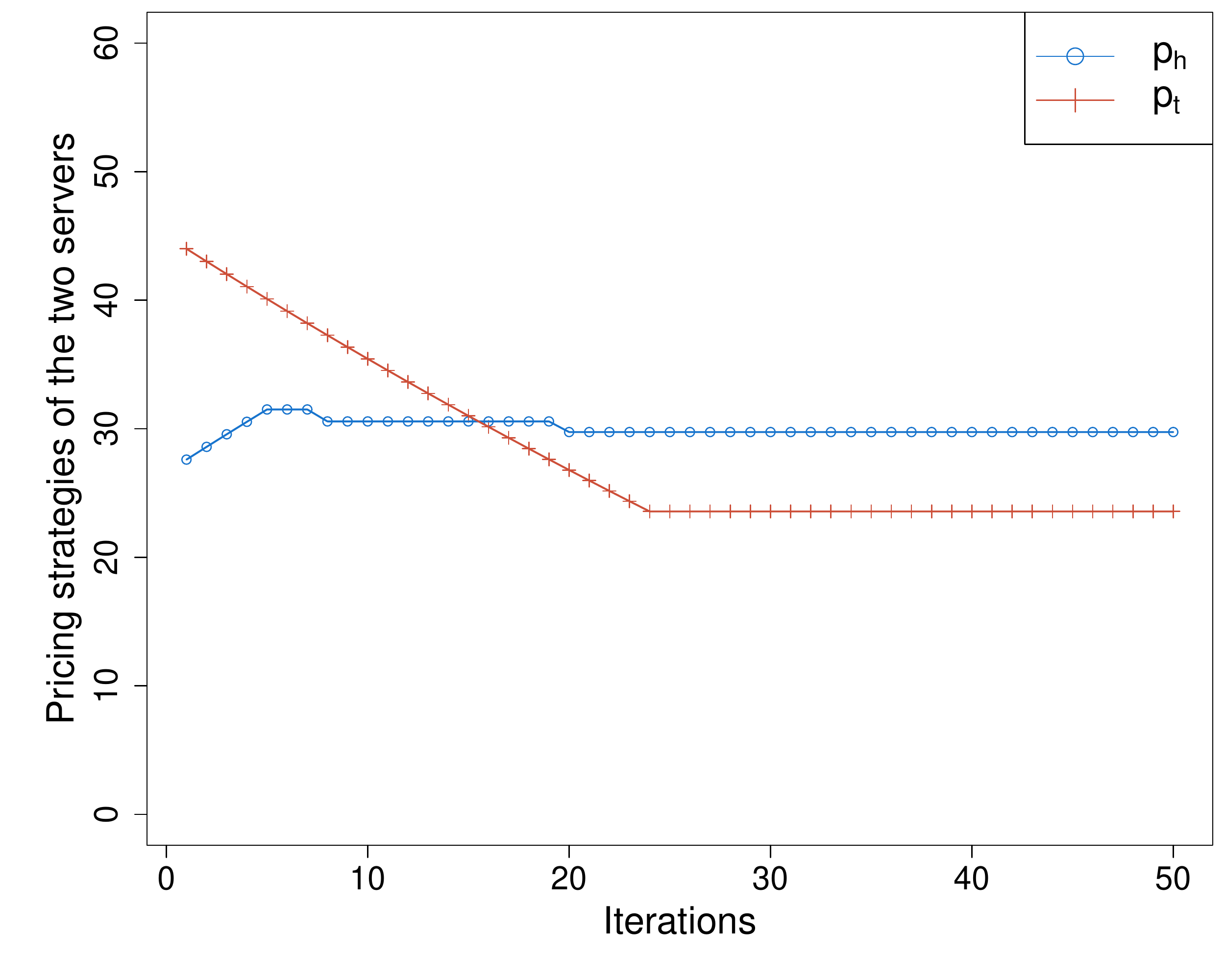}
	}
	\subfigure[Initialization:($p_h$=43.2, $p_t$=80)]{ \label{fig:itera-high}
		\includegraphics[width=.22\textwidth]{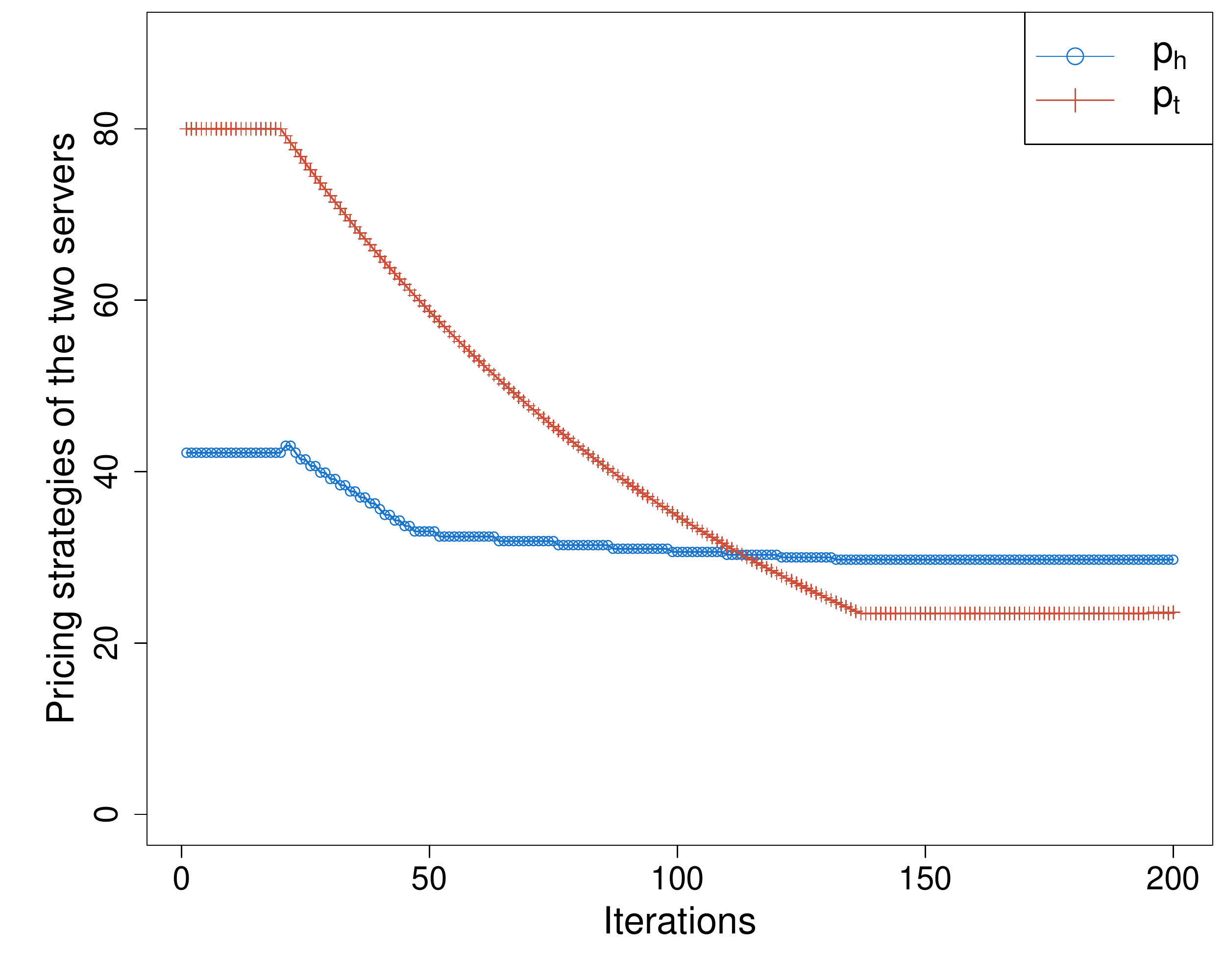}
	}
	\caption{The convergence process of $\sf{FNES}$ with different initialization.}
	\label{fig:itera-ini}
\end{figure}

\begin{figure}
	\centering
	\subfigure[$\Delta = 0.3$]{ \label{fig:itera-d3}
		\includegraphics[width=.22\textwidth]{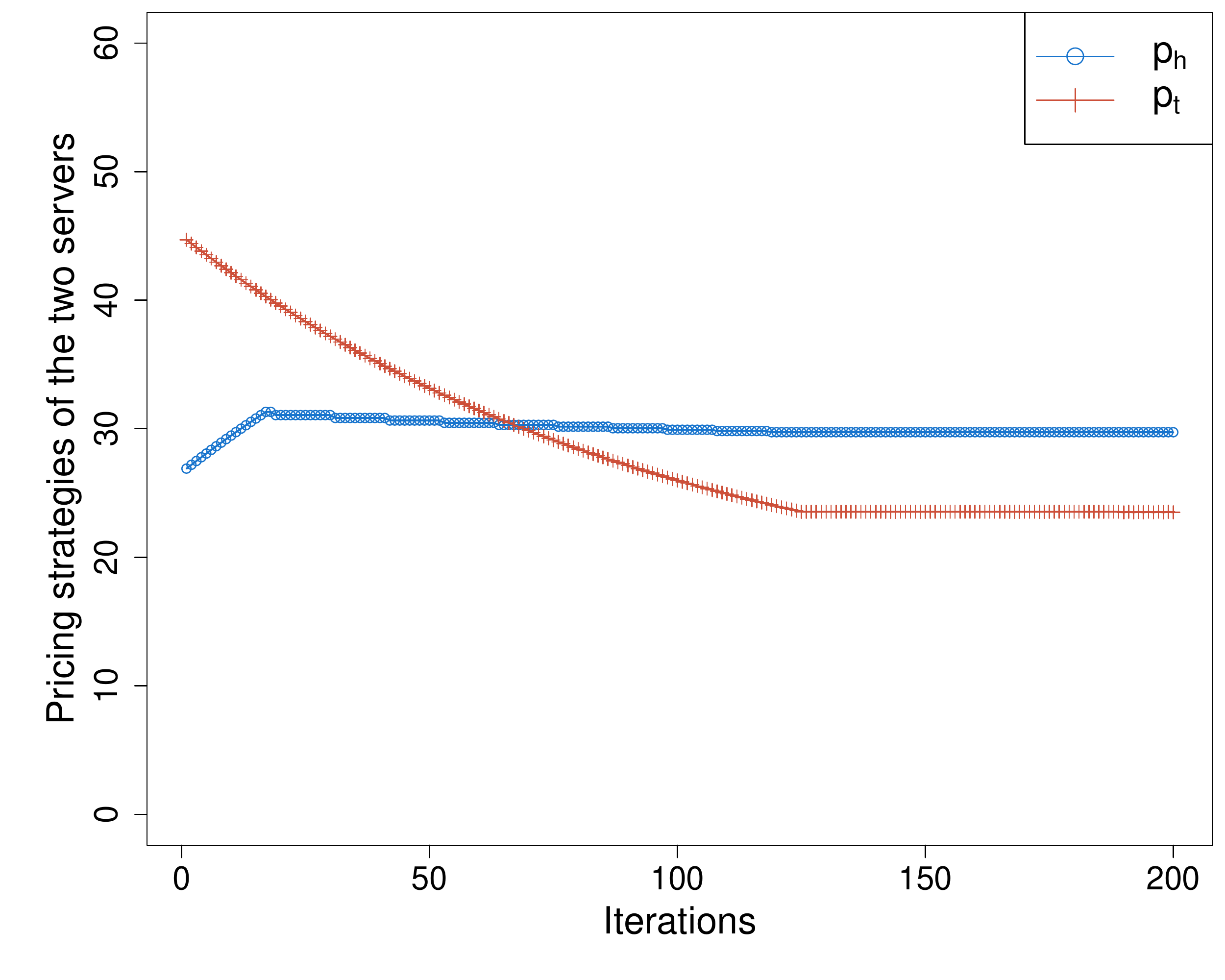}
	}
	\subfigure[$\Delta = 1$]{ \label{fig:itera-d1}
		\includegraphics[width=.22\textwidth]{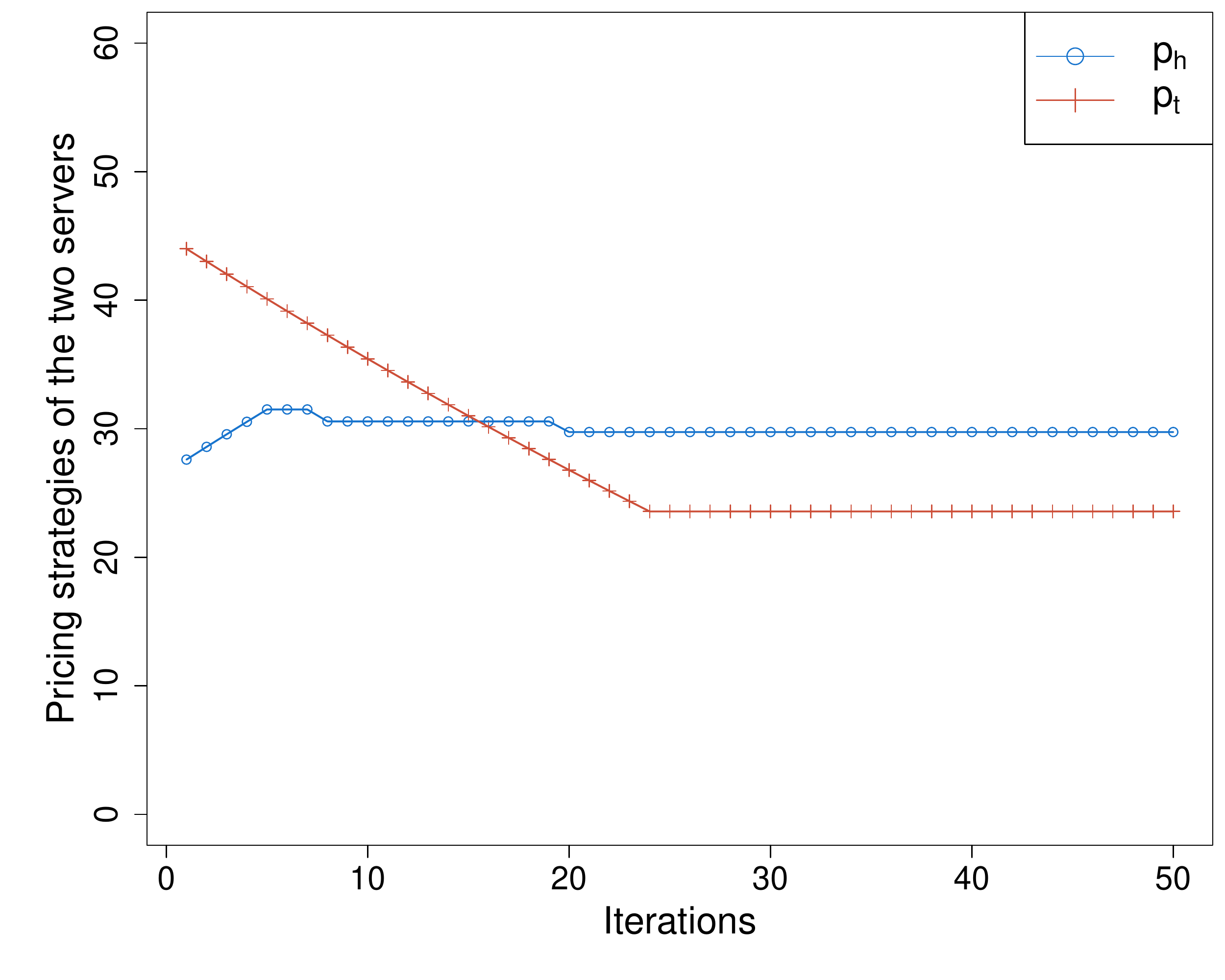}
	}
	\caption{The convergence process of $\sf{FNES}$ with different $\Delta$.}
	\label{fig:itera-delta}
\end{figure}

\begin{figure*}
	\centering
	\subfigure[]{ \label{fig:reward-price}
		\includegraphics[width=.22\textwidth]{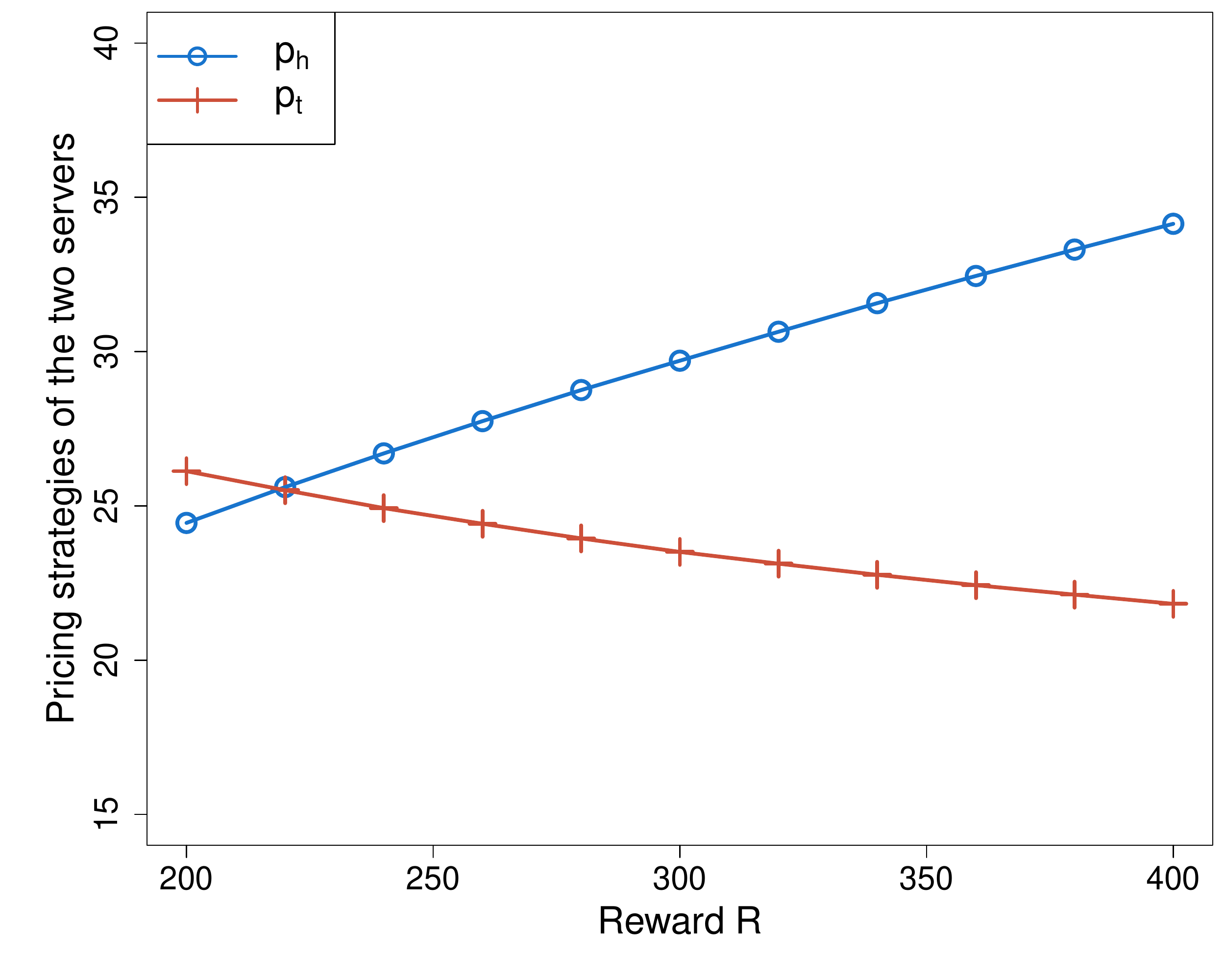}
	}
	\subfigure[]{ \label{fig:reward-utility}
		\includegraphics[width=.22\textwidth]{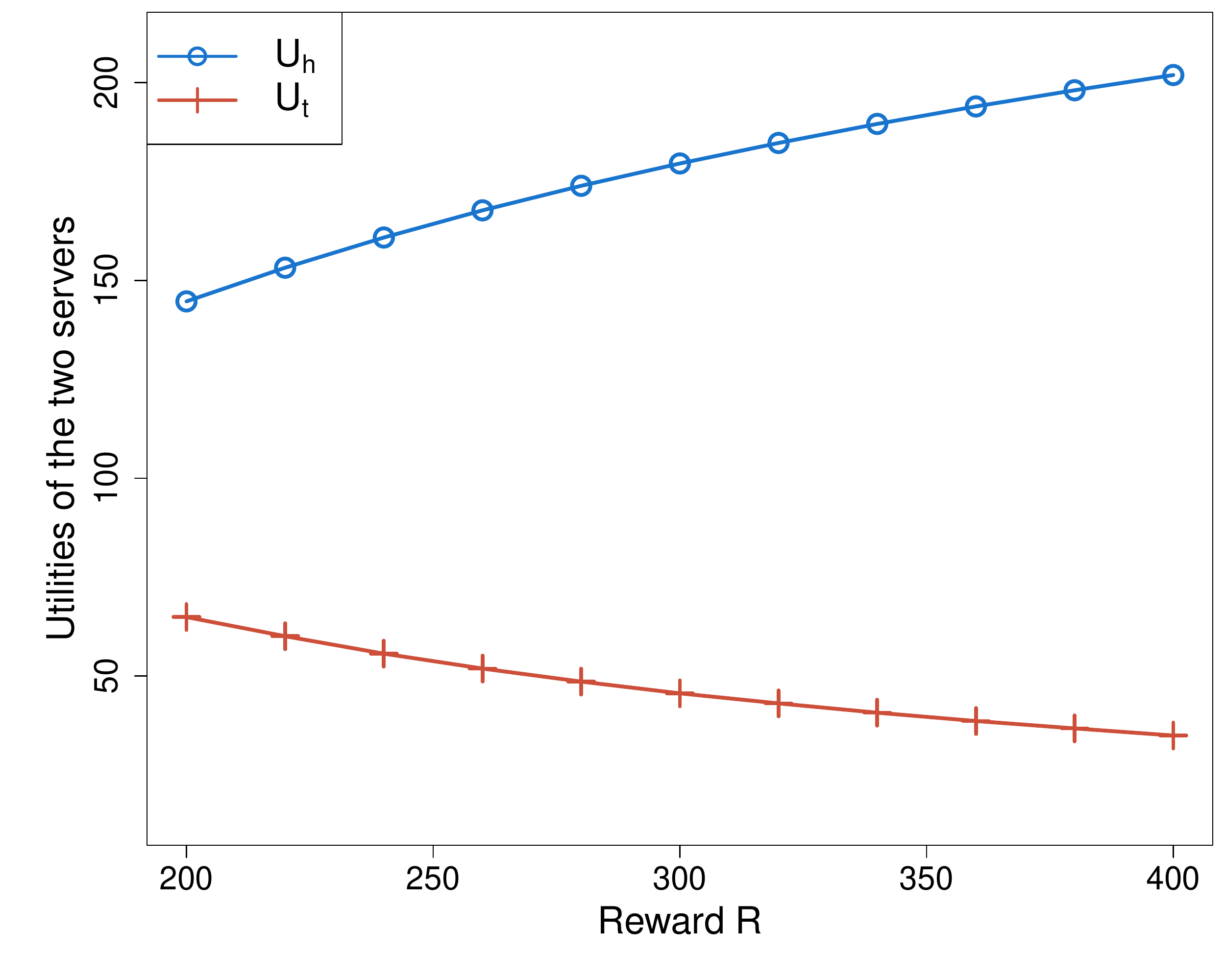}
	}
	\subfigure[]{ \label{fig:reward-power}
		\includegraphics[width=.22\textwidth]{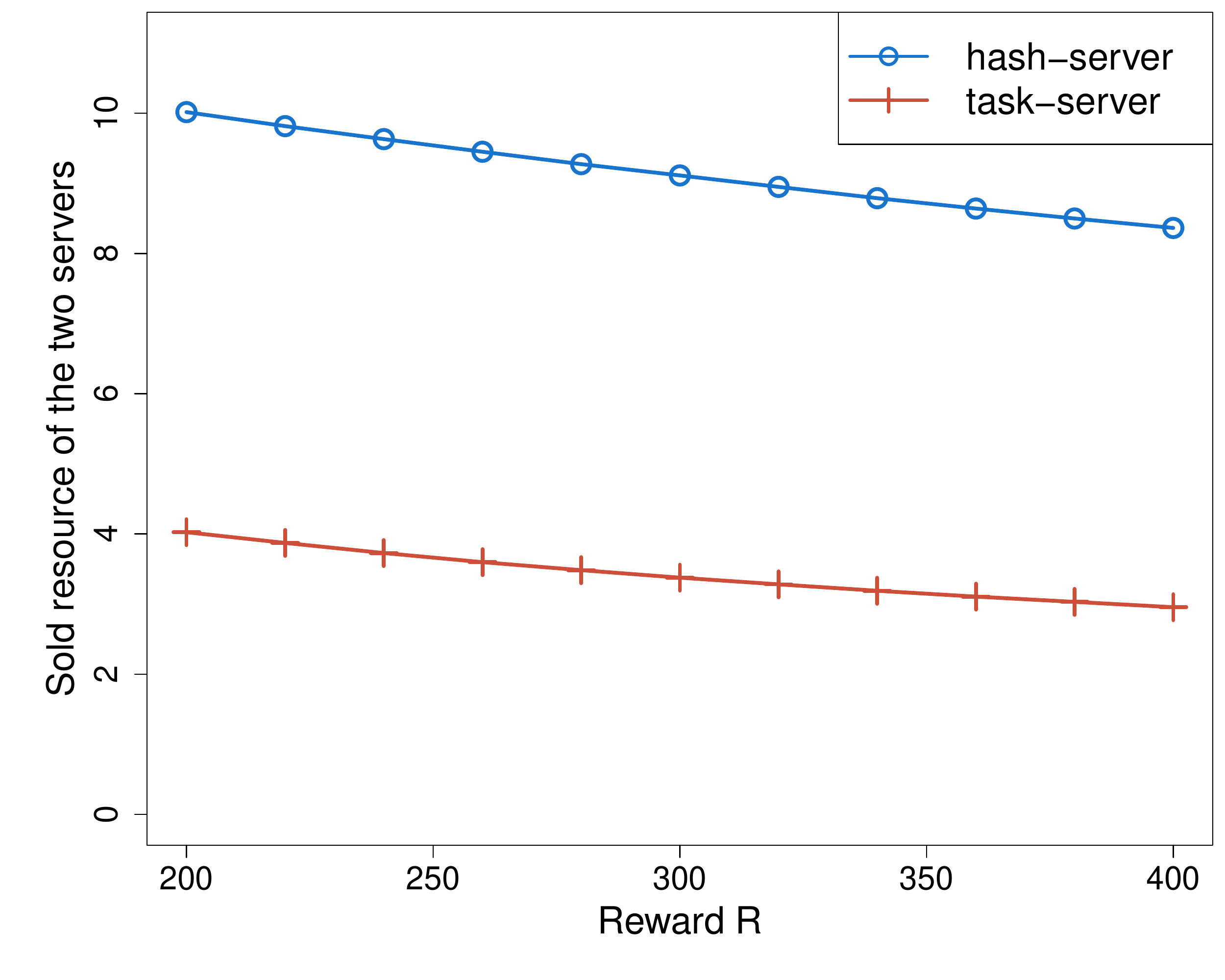}
	}
	\subfigure[]{ \label{fig:reward-profit}
		\includegraphics[width=.22\textwidth]{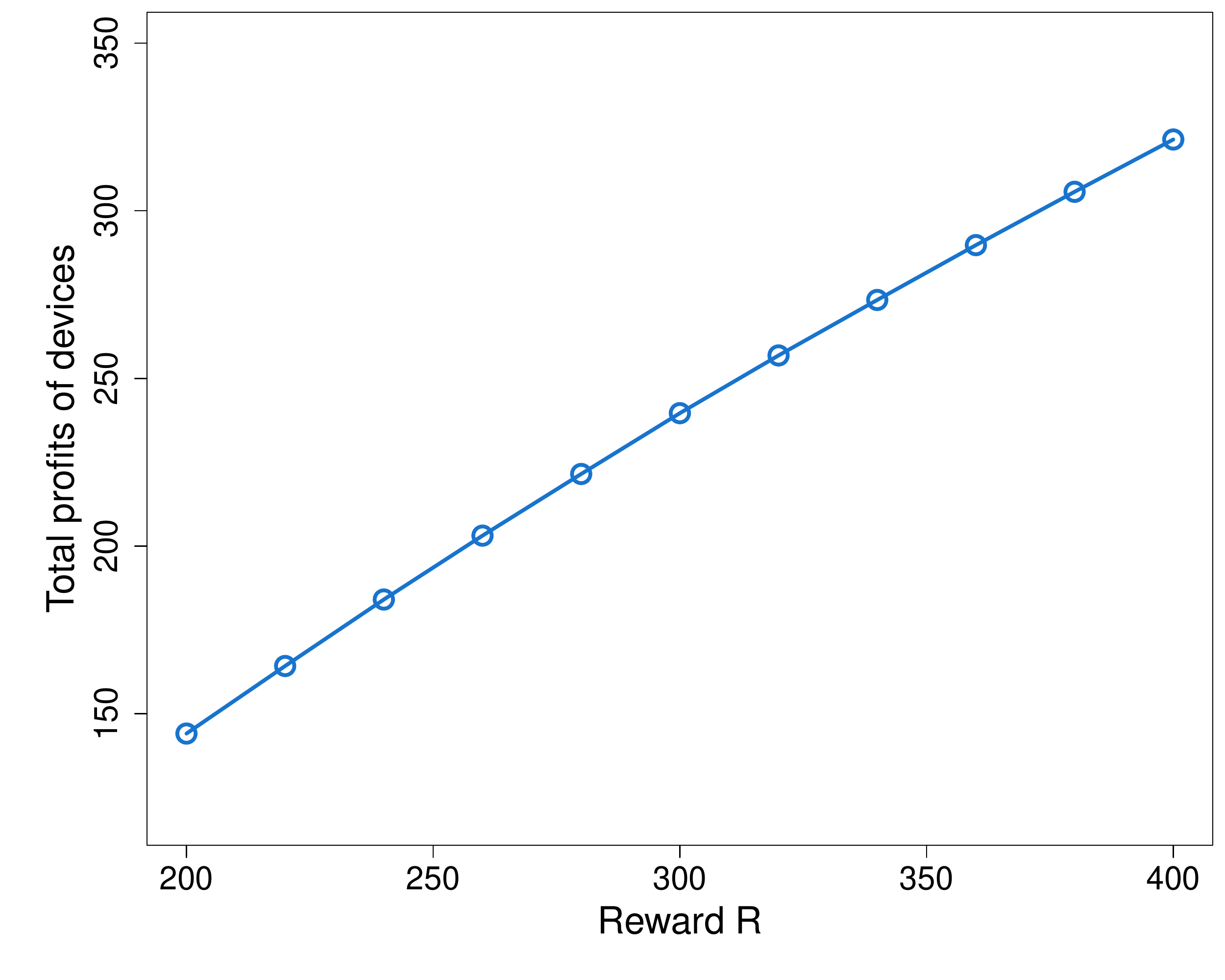}
	}
	\caption{The effect of the $R$ on (a) pricing strategies of servers, (b) utilities of servers, (c) total sold resources of servers and (d) total profits of devices.}
	\label{fig:reward-in}
\end{figure*}

\begin{figure*}
	\centering
	\subfigure[]{ \label{fig:ch-price}
		\includegraphics[width=.22\textwidth]{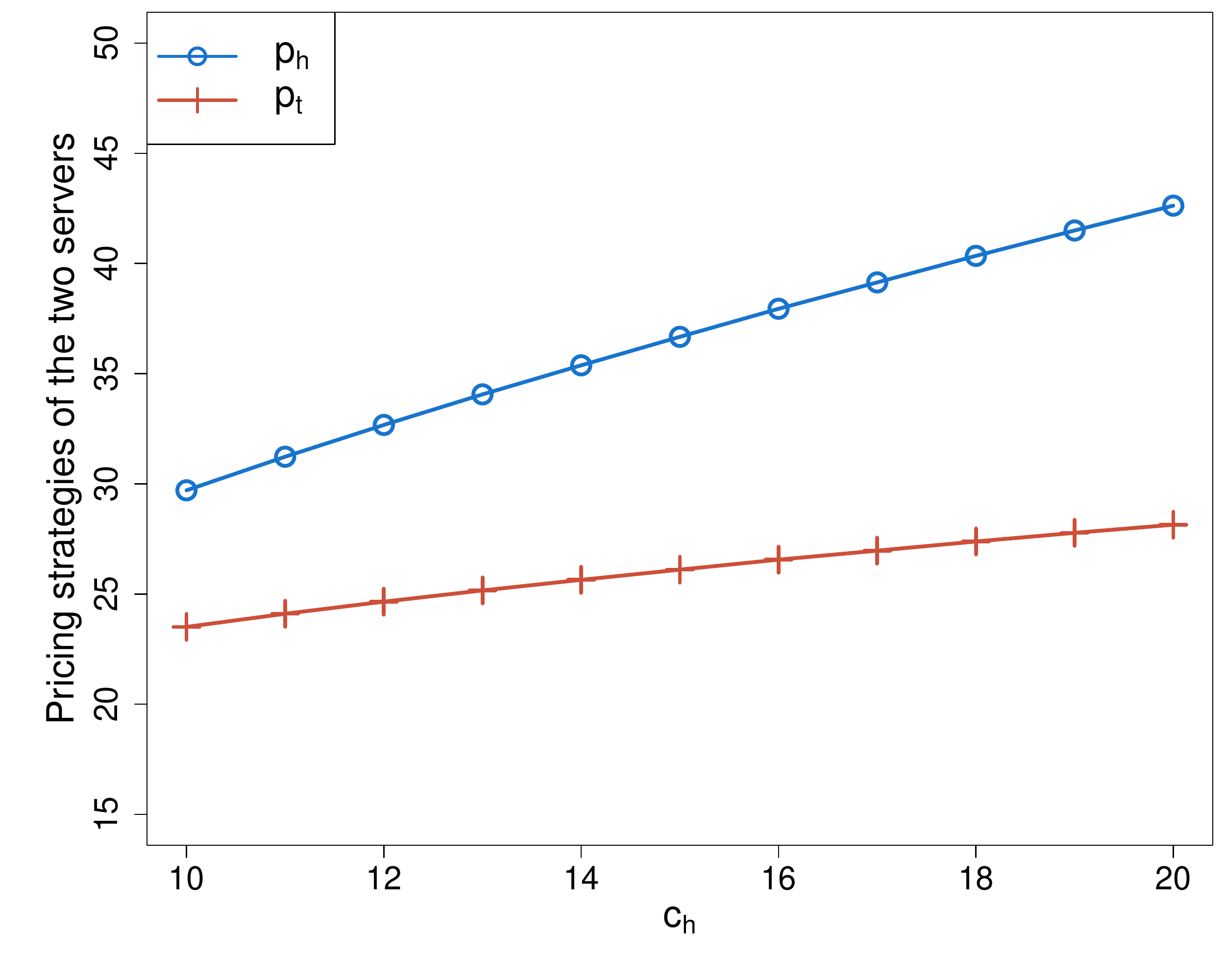}
	}
	\subfigure[]{ \label{fig:ch-utility}
		\includegraphics[width=.22\textwidth]{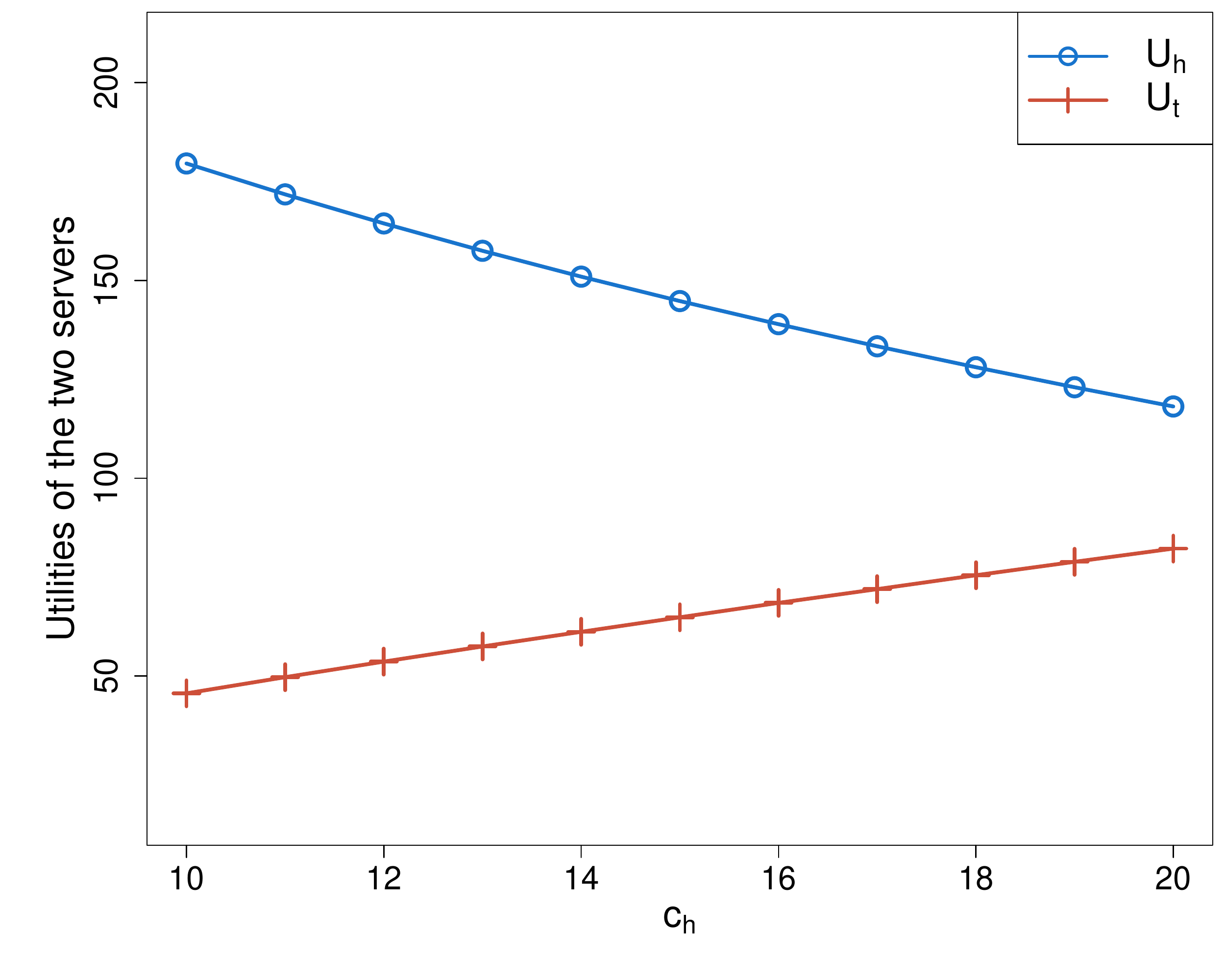}
	}
	\subfigure[]{ \label{fig:ch-power}
		\includegraphics[width=.22\textwidth]{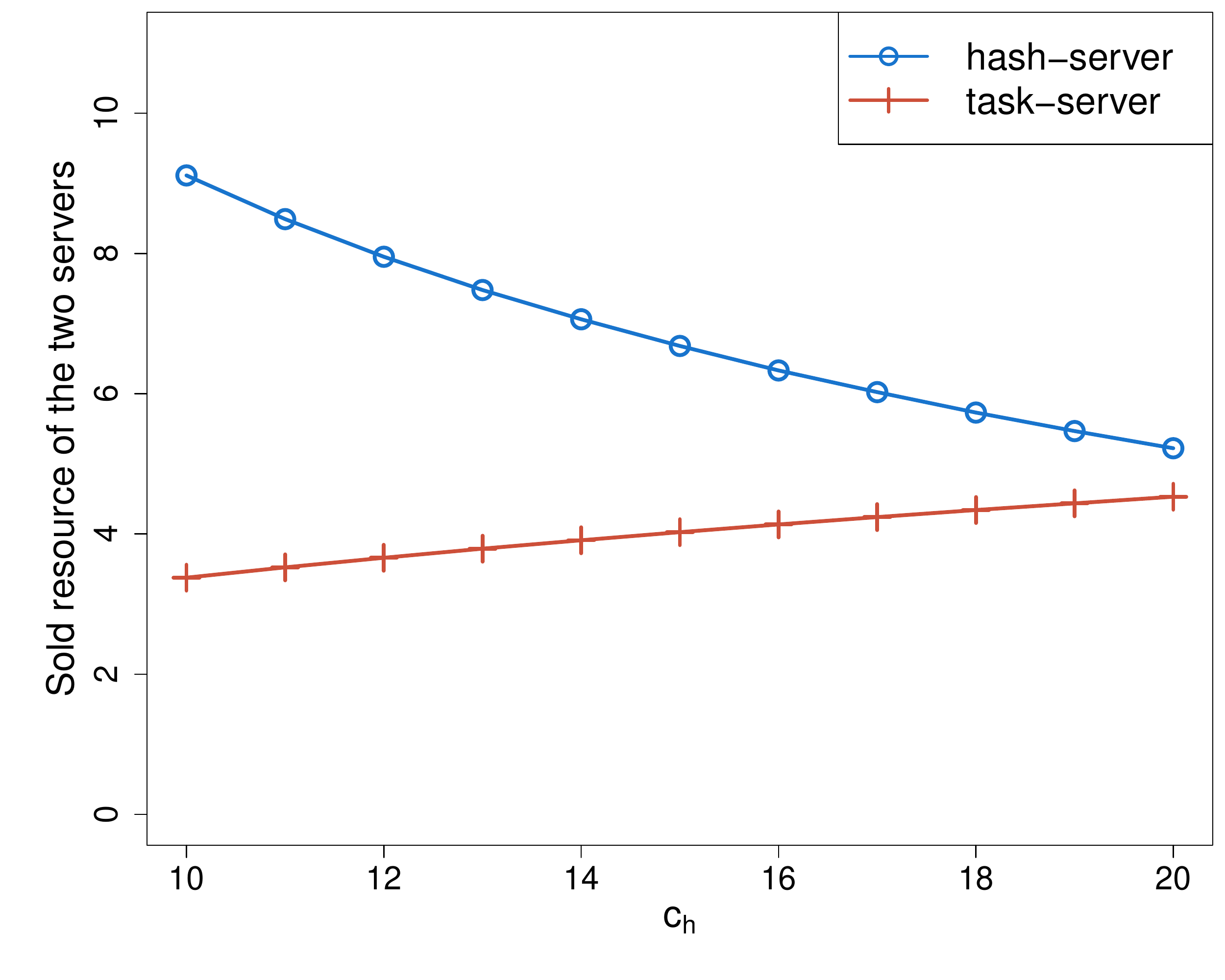}
	}
	\subfigure[]{ \label{fig:ch-profit}
		\includegraphics[width=.22\textwidth]{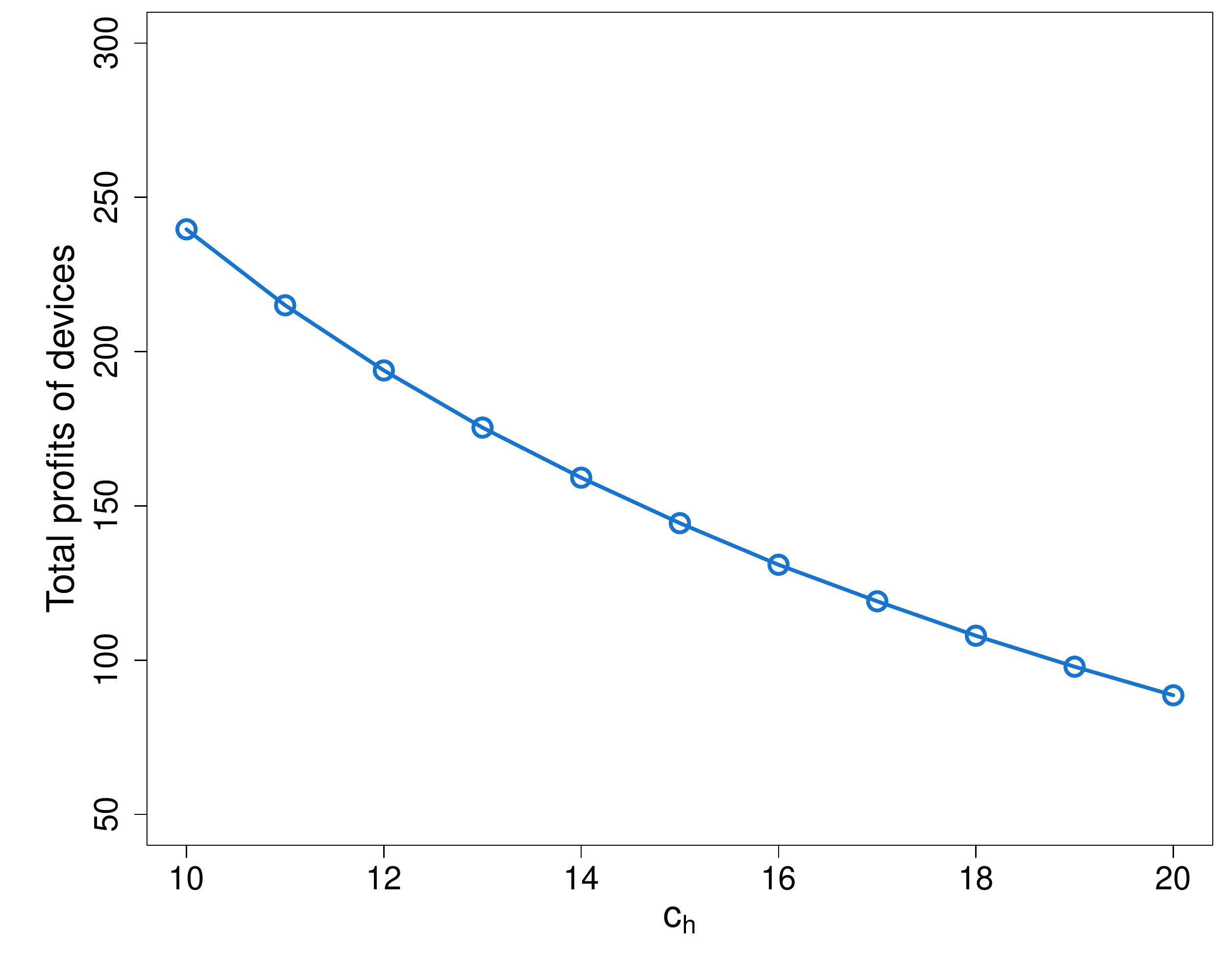}
	}
	\caption{The effect of $c_h$ on (a) pricing strategies of servers, (b) utilities of servers, (c) total sold resources of servers and (d) total profits of devices.}
	\label{fig:Ch-in}
\end{figure*}

\emph{2) The effect of the mining reward $R$:}
We investigate the effect of the mining reward $R$ on the final solution of our problem, as shown in Fig. \ref{fig:reward-in}.
When the mining reward $R$ increases from $200$ to $400$, the miners (IoT devices) will get more profit from the mining process,
and thus these devices prefer to spend their budget on purchasing hash computational power. Therefore, the hash-server
will give a higher price to get a larger utility, and the task-server has to lower its price to attract the devices. The results are shown 
in Fig. \ref{fig:reward-price} and Fig. \ref{fig:reward-utility}. 
From Fig. \ref{fig:reward-power} we can see that the total purchased hash resource of devices decreases as the reward $R$ increases.
This is because the price of the hash resource has been raised up, and the devices have limited budgets. 
It indicates that if the blockchain platform wants to attract miners to contribute more computational power by increasing the mining 
reward, it may have an opposite effect. 
Devices will get more profit as the mining reward $R$ increases, as shown in Fig. \ref{fig:reward-profit}.

\emph{3) The effect of the unit resource cost of the servers:}
We keep the unit resource cost $c_t$ of the task-server unchanged, and increase the unit resource cost $c_h$ of the hash-server 
from $10$ to $20$. 
As the unit resource cost $c_h$ raised up, the hash-server will raise its resource price to get a larger utility. Thus devices will 
allocate more budget to purchase resources from the task-server, and then the task-server will raise its resource price as it is more
competitive. The results are shown in Fig. \ref{fig:ch-price}.
The total purchased hash resource of devices will decrease due to the above reasons, meanwhile, devices will purchase more 
resource from the task-server, as shown in Fig. \ref{fig:ch-power}.
The utility of the hash-server decreases with increasing the unit resource cost $c_h$, the reason is that the total sold resources of 
the hash-server decreases as $c_h$ increases, even though the price $p_h$ has risen, the value $p_h - c_h$ almost unchanged.
The results are shown in Fig. \ref{fig:ch-utility}. 
As both of the two servers will bid a higher price as $c_h$ increases, devices will get less profit, as shown in Fig. \ref{fig:ch-profit}.
From Fig. \ref{fig:Ch-in}, we can conclude that if the server could reduce its unit resource cost, it will be more competitive and thus
obtain more benefits.

\emph{4) The effect of the budget of devices:}
If we increase the budget $b_5$ of device $s_5$ from $50$ to $190$, while the budgets of other devices keep unchanged, 
the profit gets by $s_5$ will increase because it can purchase more resources from the two servers, while the profits of
other devices decrease, as shown in Fig. \ref{fig:budget-profit}. The reason is that the increment of the budget $b_5$ will cause 
servers to raise their resource prices, which in turn reduced the amount of resources purchased by other devices, 
as shown in Fig. \ref{fig:budget-price}. 
The result indicates that the budget of devices will affect each other's profit in an indirect way.

\begin{figure}
	\centering
	\subfigure[]{ \label{fig:budget-profit}
		\includegraphics[width=.22\textwidth]{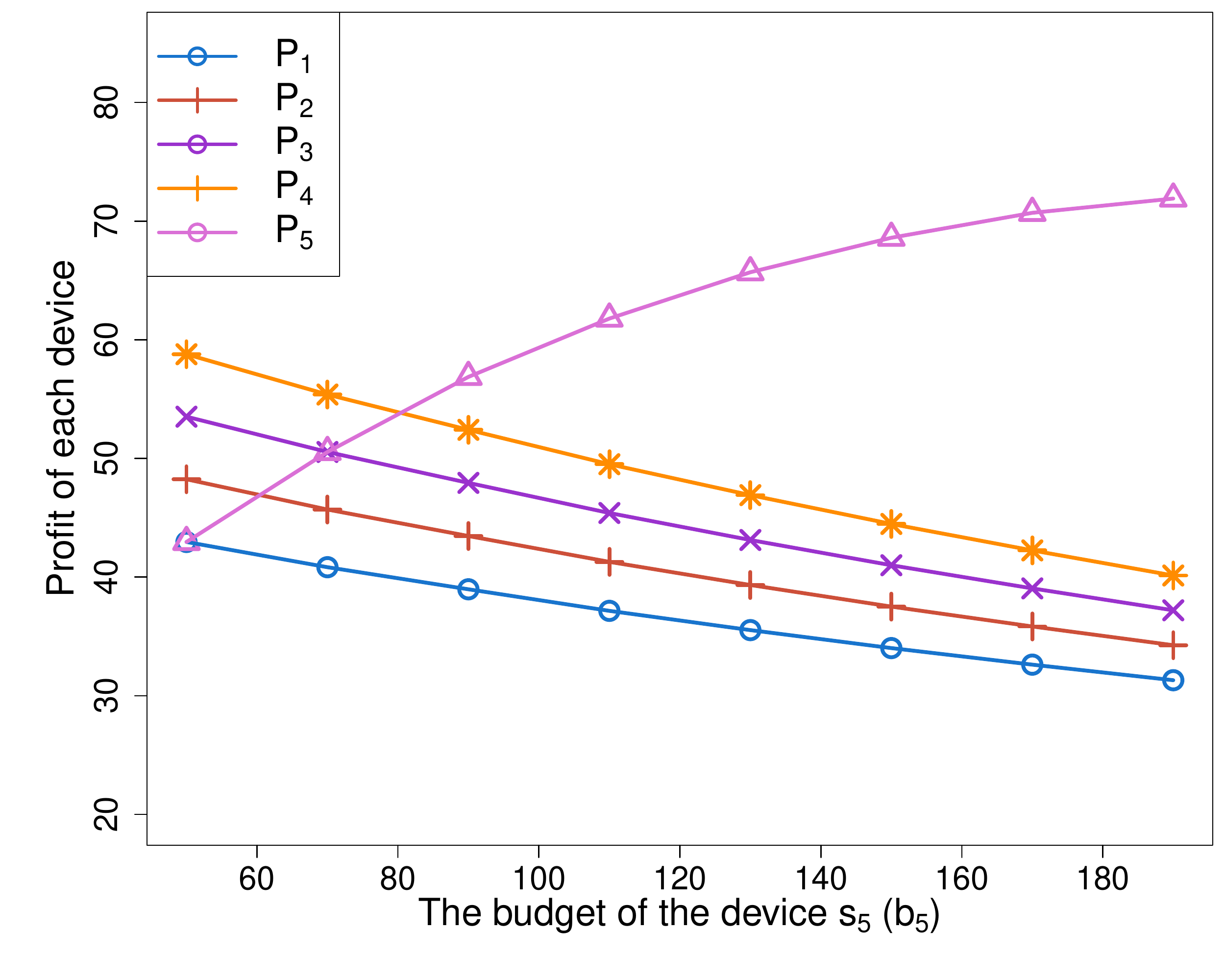}
	}
	\subfigure[]{ \label{fig:budget-price}
		\includegraphics[width=.22\textwidth]{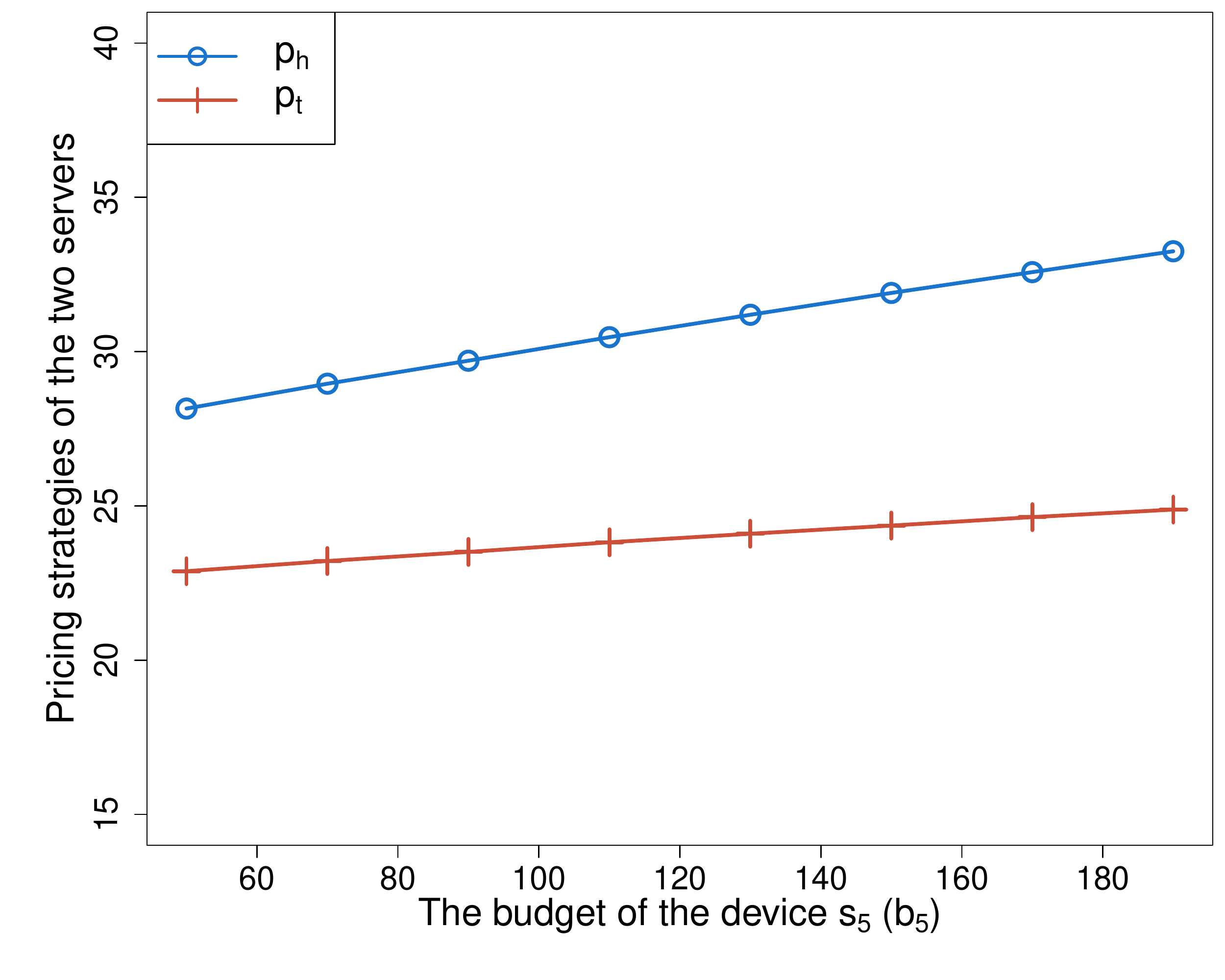}
	}
	\caption{The effect of the budget $b_5$ of device $s_5$.}
	\label{fig:budget}
\end{figure}

%%%%%%%%%%%%%%%%%%%%%%%%%%%%%%%%%%%%%%%%%%%%%%%%%%%%%%%%%%%%%%%%%%%%%
%%%%%%%%%%%%%%%%%%%%%%%%%%%%%%%%%%%%%%%%%%%%%%%%%%%%%%%%%%%%%%%%%%%%%%%%%%%%%%%%%%%%%
\section{Conclusion}\label{sec:Conclusions}
In this paper, we study the pricing and budget allocation problem between edge servers and IoT devices in an IoT blockchain network.
We first introduce the architecture of IoT blockchain with edge computing, and describe the operation of the IoT blockchain system.
Then, we model the interaction between edge servers and IoT devices as a multi-leader multi-follower Stackelberg game.
We prove the existence and uniqueness of the Stackelberg equilibrium, and design efficient algorithms to get the Stackelberg 
equilibrium point. 
Finally, we validate the correctness and effectiveness of our designs by conducting extensive simulations. 
%%%%%%%%%%%%%%%%%%%%%%%%%%%%%%%%%%%%%%%%%%%%%%%%%%%%%%%%%%%%%%%%%%%%%%%%%%%%%%%%%%%%%

\ifCLASSOPTIONcaptionsoff
  \newpage
\fi



%\bibliographystyle{IEEEtran}
\bibliographystyle{IEEEtran}
\bibliography{paper}

\begin{thebibliography}{10}
\providecommand{\url}[1]{#1}
\csname url@samestyle\endcsname
\providecommand{\newblock}{\relax}
\providecommand{\bibinfo}[2]{#2}
\providecommand{\BIBentrySTDinterwordspacing}{\spaceskip=0pt\relax}
\providecommand{\BIBentryALTinterwordstretchfactor}{4}
\providecommand{\BIBentryALTinterwordspacing}{\spaceskip=\fontdimen2\font plus
\BIBentryALTinterwordstretchfactor\fontdimen3\font minus
  \fontdimen4\font\relax}
\providecommand{\BIBforeignlanguage}[2]{{%
\expandafter\ifx\csname l@#1\endcsname\relax
\typeout{** WARNING: IEEEtran.bst: No hyphenation pattern has been}%
\typeout{** loaded for the language `#1'. Using the pattern for}%
\typeout{** the default language instead.}%
\else
\language=\csname l@#1\endcsname
\fi
#2}}
\providecommand{\BIBdecl}{\relax}
\BIBdecl

\bibitem{rehman2019cloud}
M.~Rehman, N.~Javaid, M.~Awais, M.~Imran, and N.~Naseer, ``Cloud based secure
  service providing for iots using blockchain,'' in \emph{2019 IEEE Global
  Communications Conference (GLOBECOM)}.\hskip 1em plus 0.5em minus 0.4em\relax
  IEEE, 2019, pp. 1--7.

\bibitem{roman2013features}
R.~Roman, J.~Zhou, and J.~Lopez, ``On the features and challenges of security
  and privacy in distributed internet of things,'' \emph{Computer Networks},
  vol.~57, no.~10, pp. 2266--2279, 2013.

\bibitem{kim2020p2p}
D.-Y. Kim, A.~Lee, and S.~Kim, ``P2p computing for trusted networking of
  personalized iot services,'' \emph{Peer-to-Peer Networking and Applications},
  vol.~13, no.~2, pp. 601--609, 2020.

\bibitem{alaba2017internet}
F.~A. Alaba, M.~Othman, I.~A.~T. Hashem, and F.~Alotaibi, ``Internet of things
  security: A survey,'' \emph{Journal of Network and Computer Applications},
  vol.~88, pp. 10--28, 2017.

\bibitem{kshetri2017can}
N.~Kshetri, ``Can blockchain strengthen the internet of things?'' \emph{IT
  professional}, vol.~19, no.~4, pp. 68--72, 2017.

\bibitem{nakamoto2019bitcoin}
S.~Nakamoto, ``Bitcoin: A peer-to-peer electronic cash system,'' Manubot, Tech.
  Rep., 2019.

\bibitem{divya2018iota}
M.~Divya and N.~B. Biradar, ``Iota-next generation block chain,''
  \emph{International journal of engineering and computer science}, vol.~7,
  no.~04, pp. 23\,823--23\,826, 2018.

\bibitem{bai2018state}
C.~Bai, ``State-of-the-art and future trends of blockchain based on dag
  structure,'' in \emph{International Workshop on Structured Object-Oriented
  Formal Language and Method}.\hskip 1em plus 0.5em minus 0.4em\relax Springer,
  2018, pp. 183--196.

\bibitem{ferraro2019stability}
P.~Ferraro, R.~Shorten, and C.~King, ``On the stability of unverified
  transactions in a dag-based distributed ledger,'' \emph{IEEE Transactions on
  Automatic Control}, 2019.

\bibitem{pan2018edgechain}
J.~Pan, J.~Wang, A.~Hester, I.~Alqerm, Y.~Liu, and Y.~Zhao, ``Edgechain: An
  edge-iot framework and prototype based on blockchain and smart contracts,''
  \emph{IEEE Internet of Things Journal}, vol.~6, no.~3, pp. 4719--4732, 2018.

\bibitem{morabito2018consolidate}
R.~Morabito, V.~Cozzolino, A.~Y. Ding, N.~Beijar, and J.~Ott, ``Consolidate iot
  edge computing with lightweight virtualization,'' \emph{IEEE Network},
  vol.~32, no.~1, pp. 102--111, 2018.

\bibitem{novo2018blockchain}
O.~Novo, ``Blockchain meets iot: An architecture for scalable access management
  in iot,'' \emph{IEEE Internet of Things Journal}, vol.~5, no.~2, pp.
  1184--1195, 2018.

\bibitem{gai2019permissioned}
K.~Gai, Y.~Wu, L.~Zhu, L.~Xu, and Y.~Zhang, ``Permissioned blockchain and edge
  computing empowered privacy-preserving smart grid networks,'' \emph{IEEE
  Internet of Things Journal}, vol.~6, no.~5, pp. 7992--8004, 2019.

\bibitem{guo2020combined}
J.~Guo, X.~Ding, and W.~Wu, ``A blockchain-enabled ecosystem for distributed
  electricity trading in smart city,'' \emph{IEEE Internet of Things Journal},
  pp. 1--1, 2020.

\bibitem{li2020resource}
M.~Li, F.~R. Yu, P.~Si, W.~Wu, and Y.~Zhang, ``Resource optimization for
  delay-tolerant data in blockchain-enabled iot with edge computing: A deep
  reinforcement learning approach,'' \emph{IEEE Internet of Things Journal},
  2020.

\bibitem{liu2020secure}
D.~Liu, J.~Ni, C.~Huang, X.~Lin, and X.~Shen, ``Secure and efficient
  distributed network provenance for iot: A blockchain-based approach,''
  \emph{IEEE Internet of Things Journal}, 2020.

\bibitem{qi2020cpds}
S.~Qi, Y.~Lu, Y.~Zheng, Y.~Li, and X.~Chen, ``Cpds: Enabling compressed and
  private data sharing for industrial iot over blockchain,'' \emph{IEEE
  Transactions on Industrial Informatics}, 2020.

\bibitem{lei2020groupchain}
K.~Lei, M.~Du, J.~Huang, and T.~Jin, ``Groupchain: Towards a scalable public
  blockchain in fog computing of iot services computing,'' \emph{IEEE
  Transactions on Services Computing}, vol.~13, no.~2, pp. 252--262, 2020.

\bibitem{chang2020incentive}
Z.~Chang, W.~Guo, X.~Guo, Z.~Zhou, and T.~Ristaniemi, ``Incentive mechanism for
  edge computing-based blockchain,'' \emph{IEEE Transactions on Industrial
  Informatics}, 2020.

\bibitem{yao2019resource}
H.~Yao, T.~Mai, J.~Wang, Z.~Ji, C.~Jiang, and Y.~Qian, ``Resource trading in
  blockchain-based industrial internet of things,'' \emph{IEEE Transactions on
  Industrial Informatics}, vol.~15, no.~6, pp. 3602--3609, 2019.

\bibitem{xiong2017edge}
Z.~Xiong, S.~Feng, D.~Niyato, P.~Wang, and Z.~Han, ``Edge computing resource
  management and pricing for mobile blockchain,'' \emph{arXiv preprint
  arXiv:1710.01567}, 2017.

\bibitem{xiong2018mobile}
Z.~Xiong, Y.~Zhang, D.~Niyato, P.~Wang, and Z.~Han, ``When mobile blockchain
  meets edge computing,'' \emph{IEEE Communications Magazine}, vol.~56, no.~8,
  pp. 33--39, 2018.

\bibitem{ding2020incentive}
X.~Ding, J.~Guo, D.~Li, and W.~Wu, ``An incentive mechanism for building a
  secure blockchain-based industrial internet of things,'' \emph{arXiv preprint
  arXiv:2003.10560}, 2020.

\bibitem{guo2020blockchain}
S.~Guo, Y.~Dai, S.~Guo, X.~Qiu, and F.~Qi, ``Blockchain meets edge computing:
  Stackelberg game and double auction based task offloading for mobile
  blockchain,'' \emph{IEEE Transactions on Vehicular Technology}, vol.~69,
  no.~5, pp. 5549--5561, 2020.

\bibitem{johnson2001elliptic}
D.~Johnson, A.~Menezes, and S.~Vanstone, ``The elliptic curve digital signature
  algorithm (ecdsa),'' \emph{International journal of information security},
  vol.~1, no.~1, pp. 36--63, 2001.

\bibitem{aitzhan2016security}
N.~Z. Aitzhan and D.~Svetinovic, ``Security and privacy in decentralized energy
  trading through multi-signatures, blockchain and anonymous messaging
  streams,'' \emph{IEEE Transactions on Dependable and Secure Computing},
  vol.~15, no.~5, pp. 840--852, 2016.

\bibitem{zhang2019group}
S.~Zhang and J.-H. Lee, ``A group signature and authentication scheme for
  blockchain-based mobile-edge computing,'' \emph{IEEE Internet of Things
  Journal}, vol.~7, no.~5, pp. 4557--4565, 2019.

\bibitem{koshy2020sliding}
P.~Koshy, S.~Babu, and B.~Manoj, ``Sliding window blockchain architecture for
  internet of things,'' \emph{IEEE Internet of Things Journal}, vol.~7, no.~4,
  pp. 3338--3348, 2020.

\bibitem{yang2020ldv}
W.~Yang, X.~Dai, J.~Xiao, and H.~Jin, ``Ldv: A lightweight dag-based blockchain
  for vehicular social networks,'' \emph{IEEE Transactions on Vehicular
  Technology}, vol.~69, no.~6, pp. 5749--5759, 2020.

\bibitem{rosen1965existence}
J.~B. Rosen, ``Existence and uniqueness of equilibrium points for concave
  n-person games,'' \emph{Econometrica: Journal of the Econometric Society},
  pp. 520--534, 1965.

\bibitem{boyd2004convex}
S.~Boyd, S.~P. Boyd, and L.~Vandenberghe, \emph{Convex optimization}.\hskip 1em
  plus 0.5em minus 0.4em\relax Cambridge university press, 2004.

\bibitem{zhang2016multi}
H.~Zhang, Y.~Xiao, L.~X. Cai, D.~Niyato, L.~Song, and Z.~Han, ``A multi-leader
  multi-follower stackelberg game for resource management in lte unlicensed,''
  \emph{IEEE Transactions on Wireless Communications}, vol.~16, no.~1, pp.
  348--361, 2016.

\end{thebibliography}


\begin{IEEEbiography}[{\includegraphics[width=1in,height=1.25in,clip,keepaspectratio]{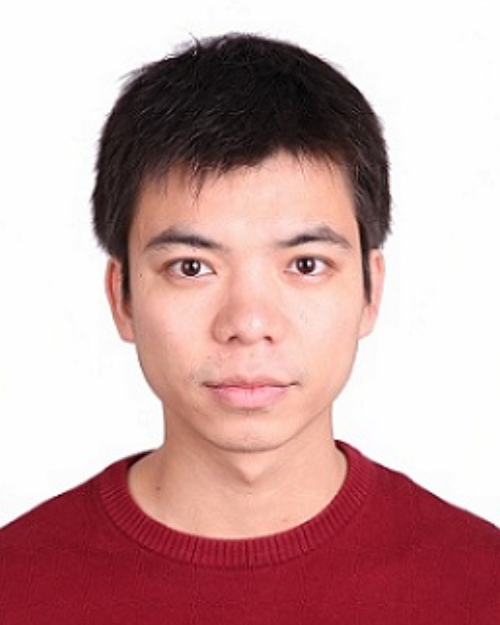}}]{Xingjian Ding}
Xingjian Ding received the BE degree in electronic information engineering from Sichuan University, Sichuan, China, in 2012. He received the M.S. degree in software engineering from Beijing Forestry University, Beijing, China, in 2017. Currently, he is working toward the PhD degree in the School of Information, Renmin University of China, Beijing, China. His research interests include wireless rechargeable sensor networks, algorithm design and analysis, and blockchain.
\end{IEEEbiography}

\begin{IEEEbiography}[{\includegraphics[width=1in,height=1.25in,clip,keepaspectratio]{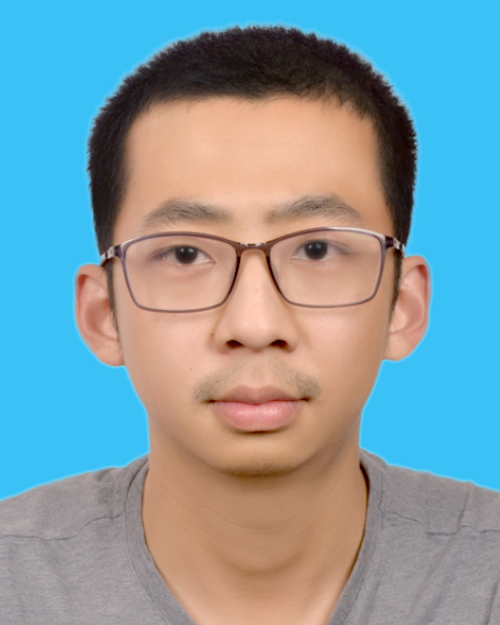}}]{Jianxiong Guo}
Jianxiong Guo is a Ph.D candidate in the Department of Computer Science at the University of Texas at Dallas. He received his BS degree in Energy Engineering and Automation from South China University of Technology in 2015 and MS degree in Chemical Engineering from University of Pittsburgh in 2016. His research interests include social networks, data mining, IoT application, blockchain, and combinatorial optimization.
\end{IEEEbiography}

\begin{IEEEbiography}[{\includegraphics[width=1in,height=1.25in,clip,keepaspectratio]{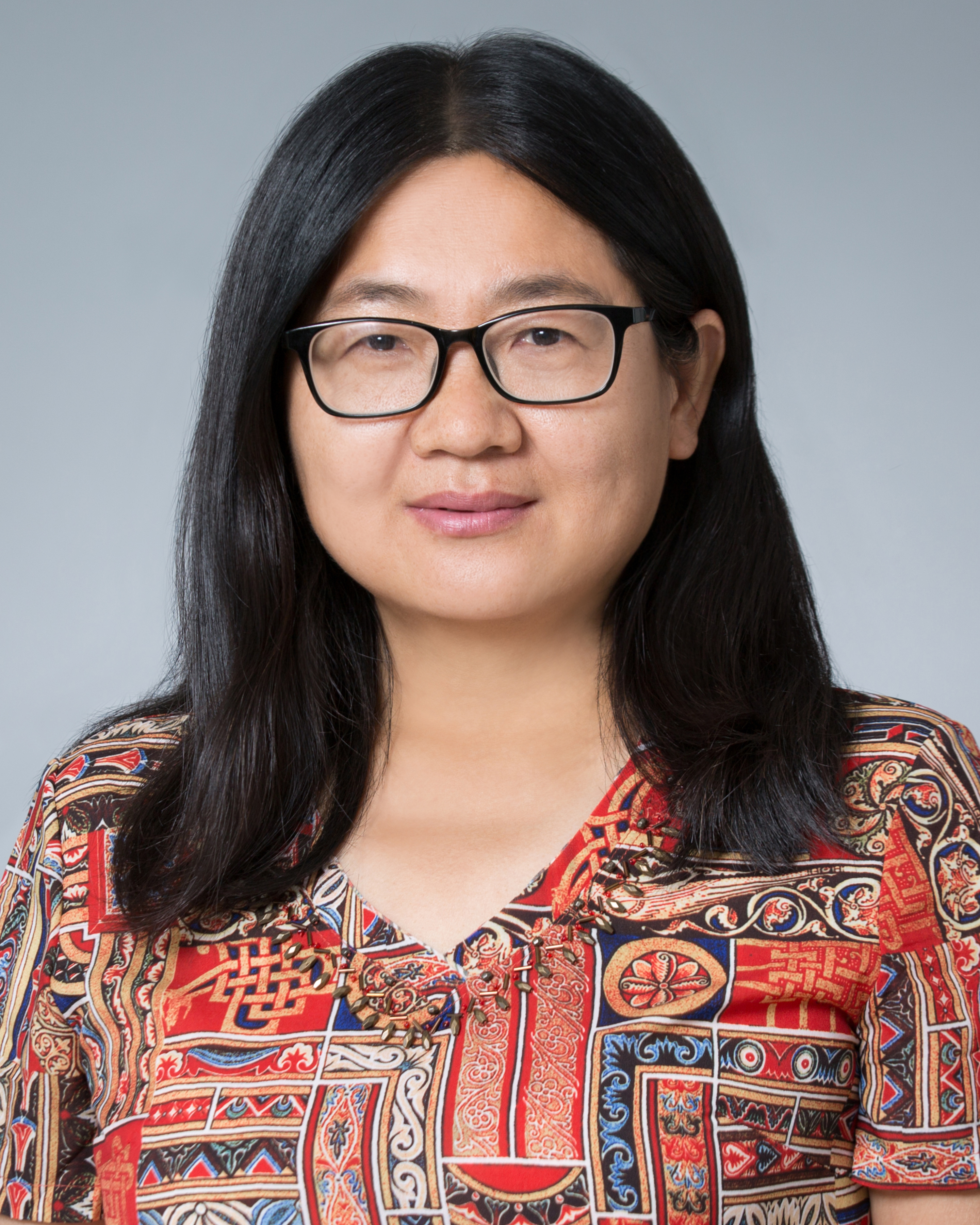}}]{Deying Li}
Deying Li is a professor of Renmin University of China. She received the B.S. degree and M.S. degree in Mathematics from Huazhong Normal University, China, in 1985 and 1988 respectively. She obtained the PhD degree in Computer Science from City University of Hong Kong in 2004. Her research interests include wireless networks, ad hoc \& sensor networks mobile computing, distributed network system, Social Networks, and Algorithm Design etc.
\end{IEEEbiography}

\begin{IEEEbiography}[{\includegraphics[width=1in,height=1.25in,clip,keepaspectratio]{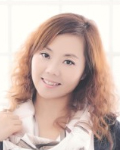}}]{Weili Wu}
Weili Wu received the Ph.D. and M.S. degrees from the Department of Computer Science, University of Minnesota, Minneapolis, MN, USA, in 2002 and 1998, respectively. She is currently a Full Professor with the Department of Computer Science, The University of Texas at Dallas, Richardson, TX, USA. Her research mainly deals in the general research area of data communication and data management. Her research focuses on the design and analysis of algorithms for optimization problems that occur in wireless networking environments and various database systems.
\end{IEEEbiography}

%\newpage


\vfill


\end{document}